\documentclass[a4paper,11pt]{article}
\usepackage[utf8]{inputenc}
\usepackage[T1]{fontenc}
\usepackage[english]{babel}
\usepackage{enumerate}
\usepackage{authblk}
\usepackage{blindtext}
\usepackage{float}
\usepackage[toc]{appendix}
\usepackage{siunitx}
\usepackage{amsfonts}
\usepackage{booktabs}
\usepackage{floatflt}
\usepackage{listings}
\usepackage{graphicx}
\usepackage{amsmath}
\usepackage{mathrsfs}
\usepackage[babel=true,tracking=true]{microtype}
\usepackage{csquotes}
\usepackage{hyperref}
\usepackage{enumitem}
\usepackage{multirow}
\usepackage{geometry}
\usepackage{tabularx}
\usepackage{subfig}
\usepackage{amsthm}
\usepackage{ amssymb }
\usepackage{physics}
\usepackage{wrapfig}
\usepackage{dsfont}
\usepackage{comment}
\providecommand*{\unit}[1]{\,\ifmmode
\mathrm{\,#1}\else\textup{#1}\fi}
\newcommand{\daga}[1]{{#1}^{\dagger}}      
\newcommand{\normt}[1]{\norm{#1}_1}
\newcommand{\Md}[1]{M_{#1}(\mathbb{C})}

\newtheorem{Definition}{Definition}
\newtheorem{Proposition}{Proposition}
\newtheorem{Lemma}{Lemma}

\newtheorem*{theorem*}{Theorem}
\theoremstyle{boldremark}
\newtheorem*{Remark*}{Remark}
\theoremstyle{boldremark}
\newtheorem{Remark}{Remark}

\theoremstyle{plain}
\newtheorem{Example}{Example}

\theoremstyle{boldremark}

\theoremstyle{plain}
\newtheorem{Corollary}{Corollary}

\begin{document}
	\date{}
	
	\title{\textbf{Open Quantum Dynamics: Memory Effects and Superactivation of Backflow of Information}}
	
	\author[1,2]{Fabio Benatti\footnote{benatti@ts.infn.it}}
	
	\author[1,2]{Giovanni Nichele}
	\affil[1]{\small \textit{Department of Physics, University of Trieste, Strada Costiera 11, I-34151 Trieste, Italy}}
	\affil[2]{\small \textit{Istituto Nazionale di Fisica Nucleare (INFN), I-34151 Trieste, Italy}}
	
	\maketitle
	
	\begin{abstract}
		We investigate the divisibility properties of the tensor products $\Lambda^{(1)}_t\otimes\Lambda^{(2)}_t$ of open quantum dynamics $\Lambda^{(1,2)}_t$ with time-dependent generators. These dynamical maps emerge from a compound open system $S_1+S_2$ that interacts with its own environment in such a way that memory effects remain when the environment is traced away.
		This study is motivated  by the following intriguing effect: one can have Backflow of Information (BFI) from the environment
		to $S_1+S_2$ without the same phenomenon occurring for either $S_1$ and $S_2$. We shall refer to this effect as the
		Superactivation of BFI~(SBFI).
	\end{abstract}


	\section{Introduction}
	An open quantum system $S$ is a system interacting with its environment in such a way that its time-evolution can be approximated by a so-called reduced dynamics. The latter is described using completely positive dynamical maps $\Lambda_t$, $t\geq 0$, on the space of states $\mathcal{S}(S)$
	that can be obtained by eliminating the environment and operating suitable approximations in order to effectively take into account its presence.
	The Markovian character, that is, the lack of memory effects, of the reduced dynamics was initially identified with $\Lambda_t$ being generated by time-independent generators $\mathcal{L}$, $\Lambda_t=\exp(t\,\mathcal{L})$, thus giving rise to one-parameter semigroups.
	In the case of bounded $\mathcal{L}$, their general structure was fully characterized by Gorini, Kossakowski, Sudarshan~\cite{GoriniKossSud} and Lindblad~\cite{LindbladTh} (GKSL generators).
	Completely positive semigroups can be rigorously obtained from a microscopic model by means of approximation techniques known as weak coupling limit~\cite{Davies}, singular coupling limit~\cite{gorini1976singular} and low-density limit~\cite{dumcke1985low}.
	In such a scenario, the dominant feature is decoherence  associated with the fact that information can only flow from the open system to its environment with no possibility of being retrieved. Decoherence is a major source of difficulties in many concrete applications such as quantum computation, quantum communication and in general quantum technologies.
	Instead, memory effects have been thought to counteract decoherence by allowing information to flow back from the environment to the system immersed in it and may thus be beneficial in many applications~\cite{WhatGoodFor}, such as quantum information processing~\cite{bylickaNonMarkovianityReservoirMemory2014}, quantum metrology~\cite{HuelgaPleniononMarkovianQuantumMetrology} and teleportation~\cite{teleportationnonmrk}.
	
	In recent years, indeed, much effort has been devoted to extending the very concept of Markovianity
	beyond the semigroup scenario (see~\cite{ChrusReview22} for a recent comprehensive review). The need for such an extension was pointed to in~\cite{BLP}, where non-Markovianity was identified with Backflow of Information (BFI) from the environment to the open system and associated with revival in the time of the distinguishability between two time-evolving~states. 
	
	In~\cite{BenattiChrusFil}, a case was presented where a dynamics $\Lambda_t$ not exposing BFI relatively to a single
	open quantum system  did instead show BFI when accompanied by an independently evolving copy of itself. In other words, it was shown that, even if $\Lambda_t$ does not show BFI, $\Lambda_t\otimes\Lambda_t$ might show it. 
	This phenomenon can be dubbed the Superactivation of Backflow of Information (SBFI).
	
	As observed above, its physical importance lies in that, by doubling a time-evolving quantum system, the decoherence effects can be diminished in the statistically coupled parties with respect to the single ones.
	The microscopic origins of SBFI, similar to the superactivation of capacity, are not yet fully explained. Indeed, they are as hard to retrieve from a microscopic system--environment interaction as BFI itself. What can be certainly ascertained, as will be explained in a forthcoming paper, is that the phenomenon is not a classical one as it is connected with the presence of non-classical correlations in the open system, not necessarily, however, unlike for the superactivation of capacity~\cite{hastingsSuperadditivityCommunicationCapacity2009,SmithQCapacity}, with entanglement.
	Moreover, a collision model derivation connects it to correlations distributed between the system and its environment~\cite{benattinichele}.
	
	In the following, we shall consider various scenarios in which SBFI occurs for local, factorized dynamics $\Lambda_t^{(1)}\otimes\Lambda_t^{(2)}$, expanding on the results of~\cite{BenattiChrusFil}.
	
	The structure of the paper is as follows. In Section \ref{sec:preliminaries}, after a survey of the two  major approaches to quantum non-Markovianity, namely, the one based on the divisibility of the dynamics and the one based on BFI, we shall review the results of~\cite{BenattiChrusFil} and extend them to the tensor products of two different dynamical maps.  Finally, in Section \ref{example:mixing}, the emergence of SBFI in  mixtures of pure dephasing  qubit dynamics will be considered, with the particular aim of investigating the stability of such a memory effect against local perturbations of one of the two dissipative evolutions.

	\section{Non-Markovianity: Divisibility and BFI}
	\label{sec:preliminaries}
	
	In what follows, we shall be concerned with one-parameter families of physically legitimate dynamical maps $\{\Lambda_t\}_{t\ge0}$, that is, with completely positive maps on the state space of a $d$-dimensional system
	generated by time-local master equations of the {form}
	\begin{equation}
		\frac{{\rm d}}{{\rm d}t}\Lambda_t=\mathcal{L}_t\circ\Lambda_t \,,
	\end{equation} where $\mathcal{L}_t$ is the time-local generator, given by $\displaystyle\mathcal{L}_t=\frac{{\rm d}}{{\rm d}t}\big(\Lambda_t\big) \circ\Lambda_t^{-1}$
	whenever $\Lambda_t^{-1}$ exists.
	
	When the generator is time-independent, $\mathcal{L}_t=\mathcal{L}$, the dynamics is a one-parameter semigroup, $\Lambda_t={\rm e}^{t\mathcal{L}}$, obeying the composition law
	$$
	\Lambda_t\circ\Lambda_s=\Lambda_s\circ\Lambda_t=\Lambda_{s+t}\ , \quad\forall s,t\geq 0\ .
	$$
	
	{Moreover}
	~\cite{BreuerPetruccione},
	$\mathcal{L}$ has the Gorini--Kossakowski--Sudarshan--Lindblad (GKSL) form
	\begin{equation}
		\label{GKSLgen}
		\mathcal{L}[\rho]=-i[H\,,\,\rho]+\sum_{i,j=1}^{d^2-1} K_{ij} \left(F_i \,\rho\, F^{\dagger}_j-\frac{1}{2}\{F_j^{\dagger}F_i,\rho\}\right)
	\end{equation}
	where $K=[K_{ij}]$ is a Hermitian positive semidefinite matrix, known as a Kossakowski matrix, and the operators $F_j$, $j\neq 0$, form a Hilbert--Schmidt orthonormal basis of traceless operators, $\textrm{Tr}(F_i^\dagger F_j)=\delta_{ij}$, with the addition  of the identity $F_{0}=\mathds{1}_d/\sqrt{d}$.
	
	If the generator is time-dependent, then the generated dynamics emerges from a time-ordered exponential and becomes a two-parameter  semigroup of maps $\Lambda_{t,s}$, $0\le s\le t$; namely, $\Lambda_{t,s}=\Lambda_{t,s_1}\circ\Lambda_{s_1,s}$, $0\le s_1\le s\le t$, where
	$$
	\Lambda_{t,s}=\mathcal{T}e^{\int_s^t{\rm d}u\,\mathcal{L}_u}={\rm id}+\sum_{k=1}^{+\infty}\int_s^t{\rm d}t_1\int_{s}^{t_1}{\rm d}t_2\cdots\int_s^{t_{k-1}}{\rm d}t_k\,\mathcal{L}_{t_1}\circ\mathcal{L}_{t_2}\circ\cdots\circ\mathcal{L}_{t_k}\ .
	$$
	
	However, the mere time dependence of the generator is no longer considered strong enough to  characterize bona fide
	non-Markovianity, that is, true memory effects.
	
	\subsection{Non-Markovianity: Lack of CP-Divisibility}
	\label{sec:CP-d}
	
	One approach to non-Markovianity is based on the notion of \textit{{divisibility}}
	\cite{RivasHuelgaPlenio,ChrusReview22}.
	
	\begin{Definition}
		\label{def:divisibility}
		The dynamics $\{\Lambda_t\}_{t\ge0}$  is called divisible if for all $t\ge s\ge 0$ there exists an intertwining two-parameter family of maps $\{\Lambda_{t,s}\}_{t\ge s\ge0 }$ such that $\Lambda_t=\Lambda_{t,s}\circ \Lambda_s$.
		If, for all $0\leq s\leq t$, $\Lambda_{t,s}$ is positive, the dynamics is called P-divisible, while, if $\Lambda_{t,s}$ is completely positive,
		the dynamics is called CP-divisible.
	\end{Definition}
	
	In~\cite{RHPmeasure}, Markovianity is identified with CP-divisibility.
	\begin{Definition}[RHP criterion]
		\label{def:CPd}
		$\{\Lambda_t\}_{t\ge 0}$ is Markovian if and only if it is CP-divisible.
	\end{Definition}
	
	The time-local generator of a one-parameter family of trace-preserving maps can always be written as
	in~\eqref{GKSLgen} with Hamiltonian $H(t)$ and Kossakowski matrix $K(t)=[K_{ij}(t)]$ both depending on time, the latter being
	only Hermitian. The GKSL characterization of completely positive semigroups fails for time-dependent generators; indeed, there can be completely positive dynamics with generators characterized by non-positive semidefinite Kossakowski matrices
	(see Example~\ref{ex:Paulimaps} below).
	Nevertheless, the GKSL characterization has the following extension to the time-dependent case which regards not the complete positivity of the generated maps, but rather their CP-divisibility~\cite{RivasHuelgaPlenio}, its indirect proof being given in Remark~\ref{remarkD1} below as a corollary of Proposition~\ref{prop:necPd}.
	
	\begin{Proposition}
		\label{prop:Cp-div}
		$\Lambda_t$ is CP-divisible if and only if $K(t)\ge 0$ $\forall \,t	\ge0$.
	\end{Proposition}
	
	Unlike for completely positive maps that are identified using their Kraus--Stinespring structure
	$\Lambda[\rho]=\sum_\alpha L_\alpha\,\rho\,L_\alpha^\dag$, the lack of a general form of only positive ones
	hampers in general the characterization of P-divisible maps that are not CP-divisible.
	The most general assertion concerning them is contained in the following lemma~\cite{ChrusReview22,Kossakowski}.
	\begin{Lemma}
		\label{chp:prop_kossak}
		The intertwining maps $\Lambda_{t,s}=\mathcal{T}e^{\int_{s}^t \dd{u} \mathcal{L}_u}$, $t\geq s\geq 0$, are positive if and only if
		\begin{equation}
			\mathcal{G}_t(\phi,\psi):=\bra{\phi}\mathcal{L}_t[\ketbra{\psi}]\ket{\phi} \ge 0\,,
			\label{chp:koss}
		\end{equation}
		for all $t\geq 0$ and $\ket{\phi}$, $\ket{\psi}$ such that $\braket{\phi}{\psi}=0$.
	\end{Lemma}
	
	\begin{Example}
		\label{ex:Paulimaps}
		The above lemma permits us to control the P-divisibility of qubit dynamics consisting of the so-called
		Pauli maps which are generated by master equations of the form
		\begin{equation}
			\label{Paulimapseq}
			\partial_t\rho_S(t)=\mathcal{L}_t[\rho_S(t)]=\frac{\lambda}{2}\,\sum_{\alpha=1}^3\gamma_\alpha(t)\,\Big(\sigma_\alpha\,\rho_S(t)\,\sigma_\alpha\,-\,\rho_S(t)\Big)\ .
		\end{equation}
		with $\lambda>0$ and $\sigma_\alpha$, $\alpha=1,2,3$ the Pauli matrices. The generator is of the form~\eqref{GKSLgen} with no Hamiltonian contribution, $d=2$, $F_\alpha=\sigma_\alpha/\sqrt{2}$  and a diagonal, time-dependent  Kossakowski matrix given by the so-called rates $\gamma_\alpha(t)$. Because of Proposition~\ref{prop:Cp-div}, then, $\gamma_\alpha(t)\geq 0$, $\alpha=1,2,3$,  is equivalent to the CP-divisibility of the Pauli maps. Moreover, $\gamma_\alpha(t)\ge 0$ is sufficient but not necessary in order to have a legitimate completely positive Pauli dynamics. \par
		The P-divisibility of the latter maps can be fully characterized by using~\eqref{Paulimapseq} in~\eqref{chp:koss}, yielding
		\begin{equation}
			\label{PauligenPos}
			\mathcal{G}_t(\phi,\psi)=\lambda\sum_{\alpha=1}^3\gamma_\alpha(t)\,\left|\bra{\phi}\sigma_\alpha\ket{\psi}\right|^2\geq 0\ ,
		\end{equation}
		for all orthogonal $\ket{\phi},\ket{\psi}\in\mathbb{C}^2$. Choosing them to be the eigenstates of $\sigma_\alpha$, $\alpha=1,2,3$,  one obtains
		\begin{equation}
			\label{pdivcond}
			\gamma_\beta(t)+\gamma_\delta(t)\geq 0\ ,\quad \beta\neq\delta\ ,\ \beta,\delta\neq \alpha\ .
		\end{equation}
		
		{These} 
		three necessary conditions are also sufficient for P-divisibility~\cite{ChrusWud2013,ChrusWud2015}; indeed, the generators $\mathcal{L}_t$ in~\eqref{Paulimapseq} at different times commute so that, from
		$$
		\mathcal{L}_t[\sigma_\alpha]=-\lambda(\gamma_\beta(t)+\gamma_\delta(t))\,\sigma_\alpha\ ,\quad \beta\neq\delta\ ,\ \beta,\delta\neq \alpha
		\ ,
		$$
		one obtains the dynamical maps
		\begin{equation}
			\label{Paulimaps}
			\Lambda_t[\sigma_\alpha]=\exp\left(-\lambda\int_{0}^{t}{\rm d}\tau\,(\gamma_\beta(\tau)+\gamma_\delta(\tau))\right)[\sigma_\alpha]
		\end{equation}
		and thus the intertwiners
		\begin{equation}
			\label{Pauliintertw}
			\Lambda_{t,s}[\sigma_\alpha]=\exp\left(-\lambda\int_{s}^{t}{\rm d}\tau\,(\gamma_\beta(\tau)+\gamma_\delta(\tau))\right)[\sigma_\alpha]\ .
		\end{equation}
		
		{Therefore},
		the conditions~\eqref{pdivcond} enforce the exponential decay of the Pauli matrices and thus the positivity of
		$\Lambda_{t,s}[\rho_S]$ for all qubit density matrices $\rho_S$, namely, the P-divisibility of the dynamical
		maps $\Lambda_t$, $t\geq 0$.
		
		As a simple concrete instance, let us choose $\gamma_1(t)=\gamma_2(t)=1$ and $\gamma_3(t)=\sin(\omega\,t)$.
		Then, given an initial density matrix in the Bloch representation,
		$$
		\rho_S=\frac{1}{2}\left(1\,+\,r_1\,\sigma_1\,+\,r_2\,\sigma_2\,+\,r_3\,\sigma_3\right)\ ,
		$$
		where $(r_1,r_2,r_3)\in\mathbb{R}^3$ with norm not larger than $1$,
		one finds
		\begin{eqnarray*}
			\Lambda_t[\rho_S(t)]&=&\frac{1}{2}\left(1+\mu(t)\,(r_1\,\sigma_1\,+\,r_2\,\sigma_2)\,+\,{\rm e}^{-2\lambda\,t}\,r_3\,\sigma_3\right)\ ,\\
			\mu(t)&=&\exp\left(-\lambda\,t-\lambda\,\frac{1-\cos(\omega\,t)}{\omega}\right)\ .
		\end{eqnarray*}
		
		{Since} $\sin(\omega\,t)$ can be negative, $\Lambda_t$ cannot be CP-divisible; however, $1+\sin(\omega\,t)\geq 0$ makes it P-divisible. To see whether $\Lambda_t$ represents a physically legitimate, that is, completely positive, evolution, one checks the positivity of the associated Choi matrix,
		\begin{eqnarray}
			\label{Choimat1}
			X_t&:=&\Lambda_t\otimes \mathrm{id}_2[P_2^+]=\frac{1}{4}\left(\mathds{1}_4+\mu(t)(\sigma_1\otimes\sigma_1-\sigma_2\otimes\sigma_2)+e^{-2\lambda\,t}\sigma_3\otimes \sigma_3 \right)\\
			&=&\label{Choimat2}
			\frac{1}{4}\left(
			\begin{array}{cccc}
				1+e^{-2 \lambda  t} & 0 & 0 & 2\mu (t) \\
				0 & 1-e^{-2 \lambda  t} & 0 &
				0 \\
				0 & 0 & 1-e^{-2 \lambda  t} &
				0 \\
				{2\mu (t)} & 0 & 0 &  1+e^{-2
					\lambda  t} \\
			\end{array}
			\right)\ ,
		\end{eqnarray}
		where $P_2^+=\ketbra{\psi_+}$ projects onto the entangled state $\ket{\psi_+}=\frac{1}{\sqrt{2}}\big(\ket{00}+\ket{11}\big)$, where $\ket{0}$ and $\ket{1}$ are eigenstates of $\sigma_3$ with eigenvalues $\pm1$. Then, $\Lambda_t$ is CP iff $X_t \ge 0$~\cite{Choi75}. 
		From~\eqref{Choimat2}, the positivity of the Choi matrix corresponds to
		\begin{equation}\label{ex1:CPcondition}
			1+e^{-2\lambda \,t}\,\geq\, 2\mu(t)\Leftrightarrow e^{\lambda \,t}+e^{-\lambda \,t}\geq 2\,e^{-\lambda/\omega}\,\exp\Big(\frac{\lambda}{\omega}\cos(\omega t)\Big)\ .
		\end{equation}
		
		{If} $\omega\ge0$, then the
		right-hand side of the second inequality above is always less than or equal to 2 and CP is guaranteed.
		Instead, for  $\omega<0$, the expansion of
		both sides for $t\to 0$ shows that \eqref{ex1:CPcondition} is violated 
		when $\lambda<|\omega|$. However, at fixed $|\omega|$, for $\lambda$ sufficiently large, complete positivity is restored, as can be seen by studying the behavior of both sides of inequality~\eqref{ex1:CPcondition}.

		\end{Example}

		In~\cite{BenattiChrusFil}, the following result was proved: it regards the case of two parties of the same type $S$, \emph{{both}
		} dynamically evolving independently under the same dynamics $\{\Lambda_t\}_{t\geq0}$.
		\begin{Proposition}
			\label{benattischrusfilippov_prop}
			$\{\Lambda_t\otimes\Lambda_t\}_{t\ge0}$ on $\mathcal{S}(S+S)$ is P-divisible if and only if $\{\Lambda_t\}_{t\ge0}$ on $\mathcal{S}(S)$ is CP-divisible.
		\end{Proposition}
		This result, which can be obtained as a corollary of Proposition~\ref{prop:necPd} below (see Remark~\ref{remarkD1}), extends to explicitly time-dependent generators $\mathcal{L}_t$ found in the case of semigroups; namely, that
		the tensor products $\Lambda_t\otimes\Lambda_t$, $t\geq 0$, where $\Lambda_t={\rm e}^{t\mathcal{L}}$, on the states of $S+S$ are positive if and only if $\Lambda_t$ is completely positive~\cite{benattifloreaniniromano2002}, a result also obtainable as a corollary of Proposition~\ref{prop:necPd} below, following the argument in Remark~\ref{remarkD1}.
		
		The physical consequences of Proposition~\ref{benattischrusfilippov_prop} are best appreciated within the context where non-Markovianity is identified using the notion of Backflow of Information (BFI).
		
		\subsection{Non-Markovianity: Backflow of Information}
		\label{sec:BFI}
		
		Differently from the RHP criterion, the BLP criterion proposed in~\cite{BLP} relates Markovianity to  the behavior under $\Lambda_t$ of the distinguishability of any two states $\rho$ and $\sigma$ of the open quantum system $S$, measured via
		\begin{equation}
			D(\Lambda_t[\rho],\Lambda_t[\sigma])=\frac{1}{2}\normt{\Lambda_t[\rho-\sigma]},
		\end{equation}
		with $\normt{\,\cdot\,}$ being the trace norm. A revival in the time of the distinguishability of two states has been interpreted in~\cite{BLP} as BFI from the environment at the roots of dissipative
		dynamics into the open quantum system and used to identify a lack of Markovian behavior. 
		
		If the two initial states are chosen with biased weights $\mu, 1-\mu$, $\mu\in[0,1]$, the proper quantifier of distinguishability becomes the trace norm of the \emph{Helstrom matrix} ${\Delta_\mu(\rho,\sigma)=:\mu\rho-(1-\mu)\sigma}$, $\rho, \sigma \in \mathcal{S}(S)$. One can then state the following~\cite{GTDWissmannBreuerAmato}:
		
		\begin{Definition}[BLP criterion]
			\label{GTDbackflow}
			The dynamics  $\{\Lambda_t\}_t$ does not display BFI  if
			\begin{equation}
				\dv{t} \normt{\Lambda_t[\Delta_\mu(\rho,\sigma)]} \le 0, \quad  \forall t\ge0,
			\end{equation}
			for all $\mu\in [0,1]$ and all $\rho, \sigma$ $\in$ $\mathcal{S}(S)$, in which case it is called Markovian.
		\end{Definition}
		
		Moreover, the following result holds for invertible maps~\cite{ChrusManiscalco,GTDWissmannBreuerAmato}, which are always divisible by $\Lambda_{t,s}=\Lambda_t\circ\Lambda_s^{-1}$.
		
		\begin{Proposition}
			\label{recall2}
			P-divisible maps $\{\Lambda_t\}_{t\ge0}$ do not display BFI. Vice versa, invertible maps $\{\Lambda_t\}_{t\ge0}$ that do not display BFI are P-divisible.
		\end{Proposition}
		
		\begin{Remark}\label{rmk:Kolmogdescrease}
			The BLP criterion for Markovianity 
			is based on P-divisibility and is not equivalent to the CP-divisibility criterion. 
			Indeed, maps which are P-divisible but not CP-divisible do not display  Backflow of Information (a typical example being the ``eternally'' non-Markovian evolution first proposed in~\cite{HallCanonical}). On the contrary, CP-divisible maps cannot show Backflow of Information, for they are P-divisible.
			
			In~\cite{CKR}, the concepts of CP-divisibility and Backflow of Information were reconciled by coupling the system $S$ to
			an inert ancilla $A$ of the same dimension and studying the information flow under the dynamics $\Lambda_t\otimes\mathrm{id}_d$ of the compound $S+A$ evolution. It was shown that an invertible dynamics $\{\Lambda_t\}_{t\geq 0}$  is CP-divisible if and only if
			$$
			\dv{t} \normt{\Lambda_t	\otimes\mathrm{id}_d [\Delta_\mu(\rho,\sigma)]} \le 0\ ,
			$$
			for all $\mu\in[0,1]$ and all $\rho,\sigma \in \mathcal{S}(S+A)$.
		\end{Remark}
		
		For invertible maps, Proposition \ref{benattischrusfilippov_prop} implies that the absence of BFI
		for bipartite systems of identically evolving parties enforces the CP-divisibility of $\Lambda_t$.
		\begin{Corollary}\label{corollary:SBFI}
			An invertible dynamics $\{\Lambda_t\}_{t\ge0}$ is CP-divisible if and only if no BFI occurs for $\Lambda_t\otimes \Lambda_t$, namely,
			\begin{equation}
				\dv{t} \normt{\Lambda_t\otimes \Lambda_t[\Delta_\mu(\rho,\sigma)]} \le 0,
			\end{equation}  for all $\mu\in[0,1]$ and all $\rho,\sigma\in \mathcal{S}(S+S)$.
		\end{Corollary}
			%
		Proposition \eqref{recall2} and Corollary \eqref{corollary:SBFI} moreover imply that\begingroup\makeatletter\def\f@size{9}\check@mathfonts
		\def\maketag@@@#1{\hbox{\m@th\normalsize\normalfont#1}}%
		\begin{equation} \label{sbfi}
			\Lambda_t~\textrm{P-divisible, not CP-divisible} \implies \begin{cases}\dv{t}\normt{\Lambda_t[\Delta_\mu^{(d)}]}\le0, \qquad \forall\, t\ge0, \,\forall \, \Delta_\mu^{(d)},\\
				\\
				\exists \, t>0, \Delta_\mu^{(d^2)}\quad  \textrm{s.t.} \quad \dv{t}\normt{\Lambda_t\otimes \Lambda_t[\Delta_\mu^{(d^2)}]}>0.\end{cases}
		\end{equation}
		\endgroup
		where the notation $\Delta_\mu^{(d)}$ indicates a Helstrom matrix in $\Md{d}$. We refer to \eqref{sbfi} as the \emph{{Superactivation of Backflow of Information}} (SBFI).
		
		\begin{Remark}
			\label{SBFIvs SC}
			As already illustrated in the Introduction, such a phenomenon means that in order to have BFI in a tensor product dynamics of a bipartite quantum system it is not necessary to have it in one or the other of the two parties.
			As such, it reminds us of the Capacity Superactivation~\cite{hastingsSuperadditivityCommunicationCapacity2009,SmithQCapacity} in that one can send information through the tensor product of two
			quantum communication channels that by themselves cannot transmit any information. In the same vein, we shall show below instances of the tensor products of two dynamics not exhibiting BFI that show it (as will be the case, for example, in Section \ref{example:mixing}, for the mixtures of pure dephasing qubit dynamics and their tensor products).
		\end{Remark}

		\subsection{General Tensor Families}
		\label{sec:generaltensors}
		
		Let us now consider the case of the general families of tensor product maps $\Lambda_t^{(1)}\otimes \Lambda_t^{(2)}$; namely, the two parties need not evolve in time according to the same reduced dynamics. In the semigroup case, it turns out that the maps  $\Lambda_t^{(1)}\otimes \Lambda_t^{(2)}$ can be positive
		without both single-system maps being completely positive~\cite{BenattiNONDEC}.
		We shall see that, unlike for equal dynamics, the P-divisibility of $\Lambda_t^{(1)}\otimes \Lambda_t^{(2)}$ does not require
		the CP-divisibility of $\Lambda_t^{(1,2)}$.
		
		The next result, based on~\cite{BFP2004}, deals with necessary conditions for the P-divisibility of $\Lambda_t^{(1)}\otimes \Lambda_t^{(2)}$.
		
		\begin{Proposition}
			\label{prop:necPd}
			Let us consider time-dependent generators, as in~\eqref{GKSLgen},
			\begin{equation}
				\label{necPd1}
				\mathcal{L}_t^{(\alpha)}[\rho]=-i[H^{(\alpha)}(t),\rho]+\sum_{i,j=1}^{d^2-1} K_{ij}^{(\alpha)}(t) \left(F_i \rho F^{\dagger}_j-\frac{1}{2}\{F_j^{\dagger}F_i,\rho\}\right)\ ,
			\end{equation}
			of dynamical maps $\Lambda_t^{(\alpha)}$, $\alpha=1,2$, respectively.
			If $\Lambda^{(1)}_t \otimes \Lambda^{(2)}_t$ is P-divisible, then for all invertible $V \in M_d(\mathbb{C})$,
			\begin{equation}\label{conditionNEC}
				K^{(1)}(t)+\mathcal{V^{\dagger}} K^{(2)}(t) \mathcal{V}\ge 0,\quad \forall \, t\ge 0
			\end{equation} where $\mathcal{V}=[\mathcal{V}_{ij}] \in M_{d^2-1}(\mathbb{C})$ is such that $V F_i^{\dagger} V^{-1}=\sum_{j=1}^{d^2-1}{\mathcal{V}_{ij} F_j^{\dagger}}$.
		\end{Proposition}

		\begin{Remark}
			\label{remarkD1}
			Notice that, by choosing $K^{(2)}(t)=0$, one reduces to maps of the form $\Lambda_t\otimes{\rm id}$; then, the above result regards the CP-divisibility of $\Lambda_t$ and makes Proposition~\ref{prop:Cp-div} a corollary of Proposition~\ref{prop:necPd}.
			Instead, taking $V=\mathds{1} \in M_d(\mathbb{C})$, $\mathcal{V}$ becomes the identity matrix in $M_{d^2-1}(\mathbb{C})$, and the P-divisibility of $\Lambda^{(1)}_t\otimes\Lambda^{(2)}_{t}$ implies
			\begin{equation} \label{k1+k2}
				K^{(1)}(t)+K^{(2)}(t) \ge 0.
			\end{equation}
			{It} is then Proposition \ref{benattischrusfilippov_prop} regarding the tensor product $\Lambda_t\otimes \Lambda_t$ that becomes a corollary of the previous proposition; indeed,~\eqref{k1+k2} reduces to $K(t)\ge0$. This, as already remarked, is equivalent to the CP-divisibility of the dynamics.
			
			On the other hand, if $\Lambda_t\otimes\Lambda_t$ is not P-divisible, one could restore P-divisibility by suitably changing the time-dependent Kossakowski matrix of the second-party dynamics, namely, by choosing $K_\varepsilon(t)=K(t)+\varepsilon\Gamma(t)$.
			Then,~\eqref{k1+k2} reads $K(t)+K_\varepsilon(t)=2K(t)+\varepsilon\Gamma(t)\ge0$. Therefore, if the single-system dynamics $\Lambda^{(1)}_t$
			is not CP-divisible, that is, if $K(t)$ is not positive semidefinite, then, in order to restore the P-divisibility of the tensor product dynamics, one has to seek a generator of the second-party time-evolution which is sufficiently strong.
			In practice, this means that to avoid SBFI due to $\Lambda_t\otimes\Lambda_t$ by changing the dynamics of the second party, one in general needs more than just a small perturbation of the second-party generator.
		\end{Remark}

		\begin{proof}[Proof of Proposition~\ref{prop:necPd} ]

			By Lemma \ref{chp:prop_kossak}, the intermediate map $\Lambda^{(1)}_{t,s} \otimes \Lambda^{(2)}_{t,s}$ is positive for any $t \ge s \ge 0$, if and only if
			\begin{equation}
				\label{necPd2}
				\mathcal{G}_s(\phi,\psi):=\bra{\phi}\mathcal{L}^{(1)}_s\otimes \mathrm{id}_d + \mathrm{id}_d\otimes \mathcal{L}^{(2)}_s\left[\ket{\psi}\bra{\psi}\right]\ket{\phi} \ge 0,
			\end{equation}
			for arbitrary orthogonal $\ket{\psi}, \ket{\phi} $. Let $\Phi=[\phi_{\alpha \beta}]$, $\Psi=[\psi_{\alpha \beta}]$ be the matrices whose entries are the vector's components with respect to a fixed orthonormal basis $\{\ket{\alpha}\otimes\ket{\beta}\}_{\alpha\beta}$. These matrices are Hilbert--Schmidt orthogonal, namely, $\braket{\phi}{\psi}=\Tr(\Phi^{\dagger} \Psi)=0$. Vice versa, given two Hilbert--Schmidt orthogonal $d\times d$ matrices $\Psi$ and $\Phi$, their entries can be taken as components of two orthogonal $d$-dimensional vectors $\ket{\psi}$ and $\ket{\phi}$ in $\mathbb{C}^{d^2}$ with respect to the chosen basis.
			
			The orthogonality of the vectors $\ket{\psi}$ and $\ket{\phi}$ reduces~\eqref{necPd2} to sums of products of contributions of the form
			\begin{eqnarray*}
				\bra{\phi}F_i\otimes\mathds{1}\ket{\psi}&=&\sum_{\alpha,\beta,\gamma=1}^d\Phi_{\alpha\beta}^*\Psi_{\gamma\beta}(F_i)_{\alpha\gamma}=\Tr\Big(F_i\Psi\Phi^\dagger\Big)\ , \\
				\bra{\phi}\mathds{1}\otimes F_i\ket{\psi}&=&\sum_{\alpha,\beta,\gamma=1}^d\Phi^*_{\alpha\beta}\Psi_{\alpha\gamma}(F_i)_{\beta\gamma}=\Tr\Big(F_i(\Phi^\dagger\Psi)^T\Big)
			\end{eqnarray*}
			and their conjugates, where $()^T$ denotes matrix transposition.
			Then, \eqref{necPd2} can be rewritten as
			\begin{equation}
				\label{aux1}
				0 \le \sum_{i,j=1}^{d^2-1} \bigg( K_{ij}^{(1)}(s) \, \overline{\Tr(F_i^{\dagger} \Phi \Psi^{\dagger})}\Tr(F_j^{\dagger} \Phi \Psi^{\dagger}) + K_{ij}^{(2)}(s)\, \overline{\Tr(F_i^{\dagger} (\Psi^{\dagger}\Phi )^{T})}\Tr(F_j^{\dagger}(\Psi^{\dagger}\Phi )^{T}) \bigg).
			\end{equation}
			
			{Let} $W$ be a generic traceless matrix in $M_{d}(\mathbb{C})$. Since every matrix is similar to its transposed matrix, given any invertible $V \in M_d(\mathbb{C})$, there exists an invertible $S$ such that $S(V^{-1}WV)S^{-1}=(V^{-1}WV)^{T}$. Choose $\Phi := V S^{-1}$, $\Psi^{\dagger} := (SV^{-1}W)$, so that $\Phi \Psi^{\dag}=W$ and $(\Psi^{\dag}\Phi)^T=V^{-1}WV$.
			With these choices,~\eqref{aux1} becomes
			\begin{equation*}
				0 \le \sum_{i,j=1}^{d^2-1} \bigg( K_{ij}^{(1)}(s) \, \overline{\Tr(F_i^{\dagger} W)}\Tr(F_j^{\dagger} W) + K_{ij}^{(2)}(s)\, \overline{\Tr(F_i^{\dagger} V^{-1}W V)}\Tr(F_j^{\dagger}V^{-1}WV) \bigg).
			\end{equation*}
			{Since} $ \Tr(W)=0$, it holds that $W=\sum_{i=1}^{d^2-1}w_i F_i$ and
			$\displaystyle
			\Tr(F_j^{\dagger} V^{-1}W V)=
			\sum_{k=1}^{d^2-1} \mathcal{V}_{jk} w_k=:v_j$.
			Finally, since $W$ is a generic traceless $d\times d$ matrix, the vector $\ket{w}=\sum_{i=1}^{d^2-1} w_i \ket{i}$ consisting of the components of $W$ with respect to the chosen Hilbert--Schmidt orthonormal basis $\{F_i\}_{i_1}^{d^2-1}$ spans $\mathbb{C}^{d^2-1}$.
			Therefore, the previous inequality yields:
			\begin{align*}
				0 \le \sum_{i,j=1}^{d^2-1} \big( K_{ij}^{(1)}(s) \, \overline{{w}_i} w_j+ K_{ij}^{(2)}(s)\, \overline{v_i} v_j\big) = \bra{w}K^{(1)}(s)+\mathcal{V^{\dagger}} K^{(2)}(s) \mathcal{V}\ket{w}\,,
			\end{align*}
			from which $K^{(1)}(s)+\mathcal{V^{\dagger}} K^{(2)}(s) \mathcal{V}\geq 0$ follows.
		\end{proof}
		\vskip .2cm
		
		We now look for a sufficient condition for the P-divisibility of generic tensor product dynamics $\Lambda_t^{(1)}\otimes \Lambda_t^{(2)}$.
		From \eqref{necPd1}, without loss of generality~\cite{HallCanonical}, the generators of the local maps can  be recast in a diagonal form with respect to a time-dependent family of Hilbert--Schmidt operators $F_{k}^{(\alpha)}(t)$, $\alpha=1,2$;
		furthermore, in the following Proposition, we will restrict to the case of an Hermitian Hilbert--Schmidt orthonormal basis ${F_k^{(\alpha)}(t)=\left(F_k^{(\alpha)}(t)\right)^\dagger}$.
		
		\begin{Remark}\label{rmk:onlynoisematters}
			As already remarked in the course of the proof of the previous proposition, only the terms
			$
			\sum_{i,j=1}^{d^2-1} K_{ij}^{(\alpha)}(t) \,\, F_i \,\rho\, \daga{F}_j
			$ from the time-local generators \eqref{necPd1} non-trivially contribute to $\mathcal{G}_s(\phi,\psi)$, for orthogonal $\ket{\phi},\ket{\psi}$, as one sees from \eqref{aux1} in the proof of Proposition \ref{prop:necPd}. The commutator and anti-commutator terms do not play any role. Therefore, focusing on only the dissipative part of the generator as in the subsequent proposition is no restriction.
		\end{Remark}
		
		\begin{Proposition}\label{prop:sufficientPD}
			Let the generators of $\Lambda_{t}^{(\alpha)}$, $\alpha=1,2$, be
			\begin{equation}
				\label{generatorsuffPd}
				\mathcal{L}_{t}^{(\alpha)}[\rho]=\sum_{k=1}^{d^2-1} \gamma_k^{(\alpha)}(t) \left(F_k^{(\alpha)}(t) \,\rho\, F_k^{(\alpha)}(t)-\frac{1}{2}\acomm{\left(F_k^{(\alpha)}(t)\right)^2}{\rho}\right) \,,
			\end{equation}
			diagonal with respect to a Hilbert--Schmidt basis of traceless and Hermitian operators $\left(F_k^{(1,2)}(t)\right)^\dagger=F_k^{(1,2)}(t)$ $\in\Md{d}$. Suppose that the rates $\gamma_k^{(1,2)}(t)$ are all positive semidefinite functions of time, except for at most two $\gamma_i^{(1)}(t)$ and $\gamma_j^{(2)}(t)$, and that, for all $t$ $\ge 0$, $\gamma_k^{(1)}(t)+\gamma_i^{(1)}(t) \ge 0$ for all $ \, k \ne i$ and $\gamma_k^{(2)}(t)+\gamma_j^{(2)}(t) \ge 0$ for all $k \ne j$. If
			\begin{equation}
				\label{aux2}
				\gamma_k^{(1)}(t)+\gamma_j^{(2)}(t) \ge 0 \qquad
				\textrm{and} \qquad\gamma_i^{(1)}(t)+\gamma_k^{(2)}(t) \ge 0 \,,
			\end{equation}
			for all $t\ge 0$ and all  $k=1,\dots,d^2-1$,  then $\Lambda_{t}^{(1)}\otimes\Lambda_{t}^{(2)}$ is P-divisible.
		\end{Proposition}
		The proof is reported in  Appendix~\ref{app:APPproofs}. The above assumptions in Proposition~\ref{prop:sufficientPD} are such that, while one rate for one system and one for the other are allowed to become negative, all sums of pairs of rates of either systems are instead forbidden to do so. In order to better illustrate the assumed properties of the rates, let us consider the following example.
		
		\begin{Example}
			\label{Ex:Paulimaps2}
			Let  $\Lambda_t^{(1,2)}$ be Pauli maps as in Example~\ref{ex:Paulimaps} defined by master equations with rates
			\begin{align}
				\gamma_1^{(1)}(t)&=1\,,\qquad \gamma_1^{(1)}(t)=1\,,\qquad \gamma_3^{(1)}(t)=\sin(t) \,, \\
				\gamma_1^{(2)}(t)&=1\,,\qquad\gamma_1^{(2)}(t)=1\,,\qquad \gamma_3^{(2)}(t)=-\sin(t) \,,
			\end{align}
			and a suitable $\lambda$ to ensure CP of $\Lambda_t^{(2)}$ (see Example~\ref{ex:Paulimaps}).
			Both dynamics are P-divisible; indeed, $\gamma_i^{(\alpha)}(t)+\gamma_j^{(\alpha)}(t)\ge0$, $\forall i \ne j$, $\alpha=1,2$.
			However, they are not CP-divisible for $\gamma_3^{(1,2)}(t)$, which becomes negative. 
			Nevertheless, $\gamma^{(1)}_3(t)+\gamma^{(2)}_3(t)=0$. Then, the conditions in \eqref{aux2}  are satisfied and ${\Lambda_t^{(1)}\otimes\Lambda_t^{(2)}}$ is P-divisible.
		\end{Example}

		
		Then, Propositions \ref{prop:necPd} and \ref{prop:sufficientPD} yield the following necessary and sufficient conditions for the
		P-divisibility of tensor products of qubit 
		Pauli maps.
		
		\begin{Proposition}
			\label{prop:qubit2tensordiff}
			Let $\Lambda_t^{(\alpha)}$, $\alpha=1,2$, be CPTP Pauli maps with time-local generators given by
			\begin{equation}
				\mathcal{L}_t^{(\alpha)}[\rho]=\frac{1}{2}\sum_{k=1}^3 \gamma_k^{(\alpha)}(t)\,(\sigma_k\, \rho\, \sigma_k-\rho)\,.
				\label{pauliagain}
			\end{equation}
			
			{Their} tensor product, $\Lambda_t^{(1)}\otimes \Lambda_t^{(2)}$, is P-divisible if and only if both $\Lambda_{t}^{(1,2)}$ are P-divisible and
			\begin{equation}
				\gamma_i^{(1)}(t)+\gamma_j^{(2)}(t) \ge 0\,, \qquad \forall \, t\ge0\,, \quad \forall \,\, i,j=1,2,3\,.
				\label{pdtensor_condition}
			\end{equation}
		\end{Proposition}
		
		The proof is reported in Appendix~\ref{app:APPproofs}. We summarize the divisibility properties of Pauli dynamics and the corresponding conditions on its generator in Table~\ref{table:qubitdiv}.
		
		\begin{table}[H]
			\caption{{Divisibility} 
				properties of Pauli maps: necessary and sufficient conditions.}
			\label{table:qubitdiv}
			\newcolumntype{C}{>{\centering\arraybackslash}X}
			\begin{tabularx}{\textwidth}{CC}
				\toprule
				$\Lambda_t$ CP-d& $\gamma_i(t)\ge0$ \\
				
				$\Lambda_t$ P-d& $\gamma_i(t)+\gamma_j(t)\ge0$, $i \ne j$\\
				
				$\Lambda_t\otimes \Lambda_t$ P-d& $\gamma_i(t)\ge0$ \\
				$\Lambda_t^{(1)}\otimes \Lambda_t^{(2)}$ P-d
				&
				$\Lambda_t^{(1,2)}$ P-d, $\gamma_i^{(1)}(t)+\gamma_j^{(2)}(t)\ge0$  \\
				\bottomrule
			\end{tabularx}
			
	\end{table}

	It is worth stressing the difference between the time-dependent case and the time-independent regime; in the latter case, CP- and P-divisibility are equivalent to the complete positivity, respectively, and the positivity of Pauli maps with rates $\gamma^{(\alpha)}_k(t)=\gamma^{(\alpha)}_k$ constant in time. 
	In such a case, it cannot be that both maps are positive but not completely positive; indeed, if it were so, there surely exist two negative rates, say, $\gamma^{(1)}_i$ and $\gamma^{(2)}_j$. Then, their sum is also negative, contradicting Equation~\eqref{pdtensor_condition}.
	As seen in Example~\ref{Ex:Paulimaps2}, in the time-dependent case, both rates $\pm\sin(t)$ can change sign without spoiling the P-divisibility of the tensor product.
	
	\begin{Remark}
		If $\Lambda_t$ is P-divisible but not CP-divisible, one can imagine keeping one party's evolution fixed and varying the second party sufficiently to achieve a  ${\Lambda}_t\otimes\widetilde{\Lambda}_t$ that is P-divisible, thereby eliminating SBFI through this variation in one of the local environments (for Pauli dynamics, this would mean varying the second-party rates until the conditions in~\eqref{pdtensor_condition} are matched). However, the map $\Lambda_t\otimes \widetilde{\Lambda}_t$ cannot be CP-divisible (since $\Lambda_t$ is not). Corollary \ref{corollary:SBFI} then implies that one can recover SBFI by doubling the system to a four-party dynamics $(\Lambda_t\otimes \widetilde{\Lambda}_t)\otimes(\Lambda_t\otimes \widetilde{\Lambda}_t)$ that would not then be P-divisible.
	\end{Remark}

	\section{Mixtures of Pure Dephasing Processes: Two-Qubit Divisibility Diagram}\label{example:mixing}
	To further illustrate the results of the previous sections, we shall now study the divisibility properties of a class of maps obtained via convex mixtures of {CP-divisible} dynamics in relation to variations in their time-local generators. The main goal is
	the characterization of a ``divisibility diagram'' of the tensor products of such mixtures.
	In particular, we are interested in mixtures of pure qubit dephasing semigroups of completely positive and trace-preserving (CPTP) maps  $\Phi_t^{(k)}=e^{t \,\mathcal{L}_k}$ arising from Pauli generators
	\begin{equation}
		\qquad \mathcal{L}_k[\rho]=\sigma_k\, \rho\, \sigma_k -\rho, \qquad k=1,2,3 \,.
		\label{simple_dephasing}
	\end{equation}
	
	{The} latter are Pauli maps as in Example~\ref{ex:Paulimaps}:
	$$
	\Phi_t^{(i)}[\sigma_\mu]=\lambda_\mu^{(i)}(t)\,\sigma_\mu\ , \quad {\lambda_\mu^{(i)}(t)=e^{-2t}, \, i \ne \mu \in \{1,2,3\}}
	\ ,\quad \lambda_i^{(i)}(t)=\lambda_0^{(i)}(t)=1\ .
	$$
	
	{Given} weights $p_k\ge0, \sum_{k=1}^3 p_k=1$, one then defines their mixtures
	\begin{equation}
		\Phi_t^{\vb{p}}:=p_1 \Phi_t^{(1)}+p_2 \Phi_t^{(2)}+p_3 \Phi_t^{(3)}\,.
		\label{mixdephasing}
	\end{equation}
	
	{These} convex combinations are also CPTP maps and satisfy $\Phi_t^{\vb{p}}[\sigma_\mu]=\lambda_\mu^{(\vb{p})}\sigma_\mu$,
	where 
	\begin{equation}
		\begin{aligned}
			\lambda_0^{\vb{p}}(t)=1\,,\qquad 	\lambda_k^{\vb{p}}(t)=&p_k+e^{-2t}(1-p_k)\ , \quad k=1,2,3\ ,
		\end{aligned}
		\label{mixeigenvalues}
	\end{equation}
	never vanish. The dynamics is thus invertible,
	$\displaystyle\left(\Phi_t^{\vb{p}}\right)^{-1}[\sigma_\mu]=\frac{1}{\lambda_\mu^{(\vb{p})}}\,\sigma_\mu$,
	and thus arises from a time-local Pauli generator $\mathcal{L}_t^{\vb{p}}=\dot\Phi_{t}^{\vb{p}}{\left(\Phi_{t}^{\vb{p}}\right)}^{-1}$; namely,
	\begin{equation*}
		\mathcal{L}_t^{\vb{p}}[\rho]=\frac{1}{2}\sum_{k=1}^3 \gamma_k^{\vb{p}}(t) \left(\sigma_k\,\rho\, \sigma_k-\rho\right)\ ,
	\end{equation*}
	where 
	the time-dependent rates $\gamma_k^{\vb{p}}(t)$ are related to the weights $\vb{p}=(p_1,p_2,p_3)$ in the following {way:} 
	\begin{subequations}\label{mixrates}
		\begin{eqnarray}
			\gamma_1^{\vb{p}}(t)=&\mu_1^{\vb{p}}(t)-\mu_2^{\vb{p}}(t)-\mu_3^{\vb{p}}(t)\,, \\
			\gamma_2^{\vb{p}}(t)=&-\mu_1^{\vb{p}}(t)+\mu_2^{\vb{p}}(t)-\mu_3^{\vb{p}}(t)\,,\\
			\gamma_3^{\vb{p}}(t)=&-\mu_1^{\vb{p}}(t)-\mu_2^{\vb{p}}(t)+\mu_3^{\vb{p}}(t)\,,
		\end{eqnarray}
	\end{subequations}
	{with}
	\begin{subequations}
		\label{g-aux}
		\begin{eqnarray}
			\label{m1}
			\mu_1^{\vb{p}}(t)&=&-\frac{p_2+p_3}{p_2+p_3+p_1\, e^{2t}}=-\frac{1-p_1}{1+p_1(e^{2t}-1)}\,,\\
			\label{m2}
			\mu_2^{\vb{p}}(t)&=&-\frac{p_1+p_3}{p_1+p_3+p_2\, e^{2t}}=-\frac{1-p_2}{1\,+\,p_2(e^{2t}-1)}\,,\\
			\label{m3}
			\mu_3^{\vb{p}}(t)&=&-\frac{p_1+p_2}{p_1+p_2+p_3\, e^{2t}}=-\frac{1-p_3}{1\,+\,p_3(e^{2t}-1)}\,.
		\end{eqnarray}
	\end{subequations}

	\subsection{One-Qubit Divisibility Diagram}\label{1qubit_dd}
	Before considering the tensor products of pairs of the above maps \eqref{mixdephasing}, we briefly review some relevant properties for the one-qubit case, which was studied in detail in~\cite{Megier_etal}.
	\begin{enumerate}[label=(\textit{\roman*}),leftmargin=2.3em,labelsep=1.mm]
		
		\item \textit{{P-divisibility.}
		} \label{prop:pd} As one can easily check from \eqref{mixrates},
		\begin{equation}\label{mixratesPd}
			\gamma_1^{\vb{p}}(t)+\gamma_2^{\vb{p}}(t)=-2\,\mu_3^{\vb{p}}(t)\ge 0\,,
		\end{equation}
		as well as for cyclic permutations of the indices. Thus (see Example \ref{ex:Paulimaps}), $\Phi_t^{\vb{p}}$ is P-divisible for all  $\vb{p}$.
		
		\item \textit{{Eternal and quasi-eternal non-Markovianity.}} 	\label{prop:eqenm} From \eqref{mixrates}, one has  \begin{equation}
			\gamma_k^{\vb{p}}(0)=2\, p_k \ge0, \qquad k=1,2,3\,.
		\end{equation}
		
		If $p_k=0$, then
		\begin{equation}
			\gamma_k^{\vb{p}}(t)<0\, \quad \forall \, t>0\,;
		\end{equation}
		for instance, letting $p_3=0\implies p_1+p_2=1$ and $t>0$,
		$$
		\gamma_3^{\vb{p}}(t)=\frac{p_2}{p_2+e^{2t}p_1}+\frac{p_1}{p_1+e^{2t}p_2}-1<p_2+p_1-1=0\,,
		$$
		while $\gamma_i^{\vb{p}}(t) > 0,\,t>0,\, i =1,2$,  as a consequence of {Property} 
		(\emph{i}).
		Such a dynamics is clearly not CP-divisible and is usually labeled as \emph{{eternally non-Markovian}} (ENM). For instance, choosing $p_1=p_2=1/2$, $p_3=0$, we end up with the ENM evolution first introduced in~\cite{HallCanonical}, with $\gamma_1(t)=\gamma_2(t)=1$ and $\gamma_3(t)=-\tanh(t)$. Similarly, if $p_1 \,p_2 \, p_3 >0$ and
		\begin{equation}
			\exists \,t^*>0 : \gamma_k^{\vb{p}}(t^*)<0  \implies \gamma_k^{\vb{p}}(t)<0  \quad \forall \,t > t^* \,,
		\end{equation}
		
		This is checked explicitly in Appendix~\ref{theappendix}. It also follows that, for a given $\vb{p}$, at most one rate $\gamma_k^{\vb{p}}(t)$ can become negative: indeed, if $\gamma_k^{\vb{p}}(t)<0$ for $t>t_k^*$ and $\gamma_j^{\vb{p}}(t)<0$ for $t>t_j^*$, $\gamma_k^{\vb{p}}(t)+\gamma_j^{\vb{p}}(t)<0$ for $t>\max\{t_k^*,t_j^*\}$, violating P-divisibility condition~\eqref{mixratesPd}.

		\item \textit{{One-qubit divisibility diagram}.} \label{prop:ddiagram1}  The above properties allow one to characterize the one-qubit ``divisibility diagram''~\cite{divdiagram2015} of the mixtures $\Phi_t^{\vb{p}}$ in terms of the parameters~$\vb{p}$. 
		Let us denote by $\mathcal{P}$ the set of weights $\{\vb{p}: p_k\ge 0, \sum_k p_k=1\}$. As remarked in {point} 
		(\emph{i}) 
		above, each $\vb{p}\in\mathcal{P}$ corresponds to a P-divisible map $\Phi_t^{\vb{p}}$. Therefore, the set of divisible maps is the union of the disjoint subsets 
		\begin{equation}
			\mathcal{CP}=\{\vb{p}: \Phi_t^{\vb{p}}\textrm{ is CP-divisible}\} \quad \hbox{and} \quad\mathcal{P}\setminus \mathcal{CP} \,,
		\end{equation}
		namely, of the subsets of weights corresponding to  CP-divisible maps, respectively, and to P-divisible but not CP-divisible maps $\Phi_t^{\vb{p}}$, respectively.
		The subset $\mathcal{CP}$ can be identified by means of {Property} 
		(\emph{ii}); indeed, if $\Phi_t^{\vb{p}}$ is not CP-divisible, 
		there must exist exactly one $k\in\{1,2,3\}$ such that ${\gamma_k^{\vb{p}}(t\to\infty)<0}$. On the other hand,
		$\Phi_t^{\vb{p}}\in\mathcal{CP}$ if and only if it has positive rates, so that:
		\begin{equation}
			\vb{p}\in\mathcal{CP} \iff \gamma_k^{\vb{p}}(t\to\infty)\ge0 \qquad   k=1,2,3\,.
		\end{equation}
		
		Assuming $p_1 \,p_2 \, p_3 >0$, the asymptotic behavior of the rates for $t\to \infty$ is as follows: 
		\begin{subequations}
			\begin{align}
				\gamma_1^{\vb{p}}(t) &\simeq \frac{e^{-2t}}{p_1\, p_2\, \,p_3}\left(-p_2\, p_3 \,(p_2+p_3)+p_1 \,p_3 \,(p_1+p_3)+p_1\, p_2\, (p_1+p_2)\right)\,,
				\label{asympt1}\\
				\gamma_2^{\vb{p}}(t) &\simeq \frac{e^{-2t}}{p_1\, p_2\, p_3}\left(p_2\, p_3 (p_2+p_3)-p_1 \,p_3 \,(p_1+p_3)+p_1\, p_2\, (p_1+p_2)\right)\,,
				\label{asympt2}\\
				\gamma_3^{\vb{p}}(t) &\simeq \frac{e^{-2t}}{p_1\, p_2\, p_3}\left(p_2\, p_3\, (p_2+p_3)+p_1\, p_3 \,(p_1+p_3)-p_1 \,p_2\, (p_1+p_2)\right)\ ,
				\label{asympt3}
			\end{align}
		\end{subequations}
		makes the region $\mathcal{CP}$ identified using the following inequalities
		\begin{subequations}\label{CPdconditions}
			\begin{align}
				p_1 \,p_2 \,(p_1+p_2) +p_2^2-p_1^2+p_1-p_2 \ge 0\,,
				\label{reg1}\\
				p_2 \,p_1 \,(p_1+p_2) +p_1^2-p_2^2+p_2-p_1 \ge 0\,,
				\label{reg2}\\
				(1+p_1 \,p_2) \,(p_1+p_2) - p_1^2-p_2^2-4 p_1\, p_2\ge0 \,,
				\label{reg3}
			\end{align}
		\end{subequations}
		where $p_3=1-p_1-p_2$ was used, making it possible to represent $\mathcal{CP}$ as a two-dimensional, triangular-like region in the $(p_1,p_2)$ plane, as in Figure \ref{fig:regions}. Conversely, maps which are only P-divisible are identified with points $\vb{p}\in\mathcal{P}\setminus\mathcal{CP}$ that violate one and only one of the inequalities~\eqref{CPdconditions}. Such maps are all examples of so-called \textit{{weakly}} non-Markovian evolutions: in fact, as stated in Proposition~\ref{recall2}, they do not display BFI. The situation is remarkably different when taking the tensor product of two such maps: in such a case, parameter regions with associated BFI appear.
	\end{enumerate}
	
	\begin{figure}[t]
			\centering
		\includegraphics[height=0.4\linewidth]{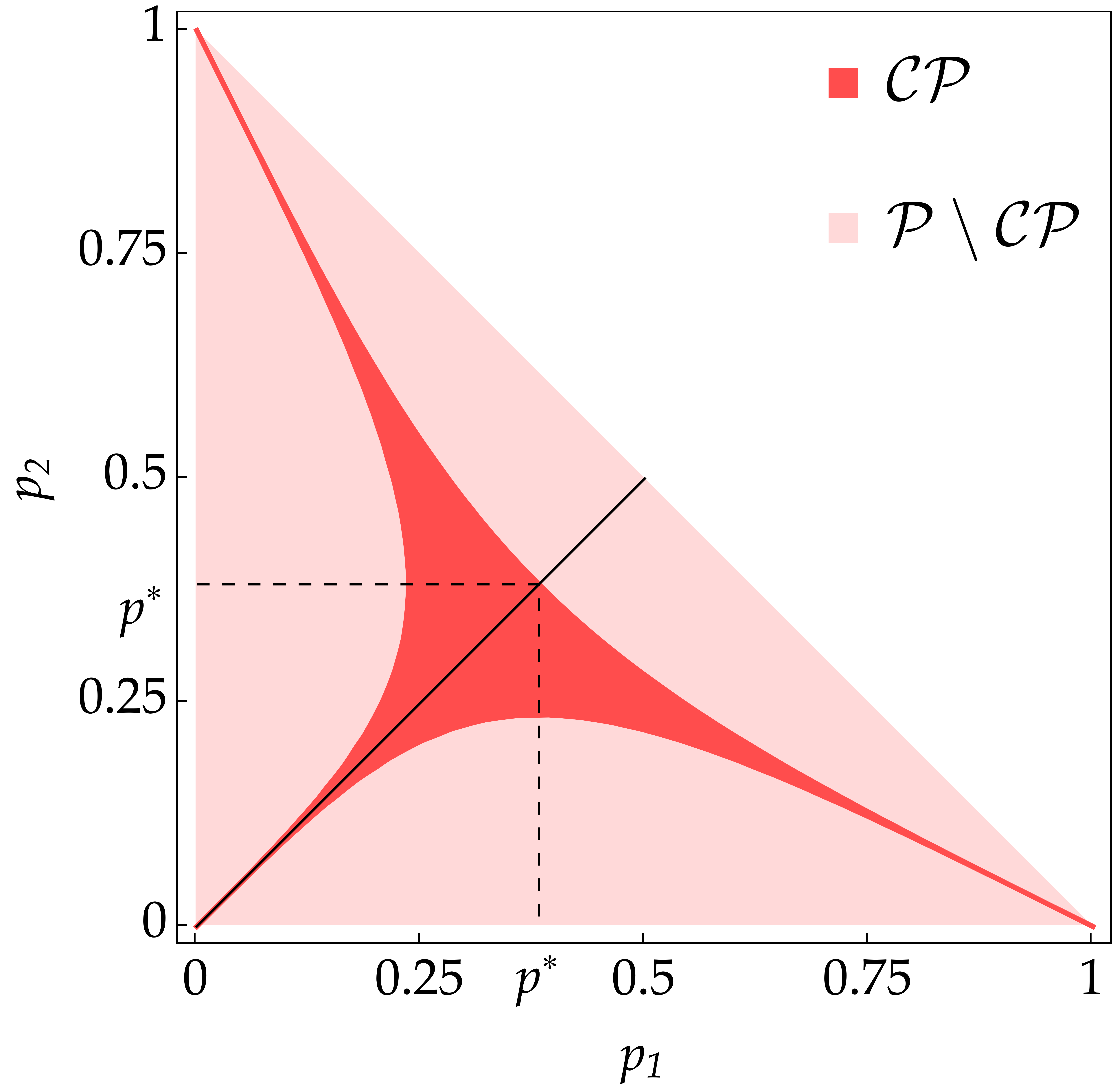}
		\caption{{One-qubit divisibility diagram. For each point $(p_{1},p_2)$, $p_3=1-p_1-p_2$. Region $\mathcal{CP}$ includes CP-divisible maps ${(\gamma_k^{\vb{p}}(t) \ge 0, \, \forall \, t\ge0)}$, while $\mathcal{P}\setminus\mathcal{CP}$ corresponds to P-divisible but not CP-divisible maps ($\exists \, k, \,t^*\ge 0 : \gamma_k^{\vb{p}}(t^*) < 0$, $\forall\, t >t^*$). The highlighted point $(p^*,p^*)$, ${p^*=	\frac{1}{2}(3-\sqrt{5})\approx0.38}$, marks the boundary between $\mathcal{CP}$ and $\mathcal{P}\setminus\mathcal{CP}$ along the line $p_1=p_2$.}}
		\label{fig:regions}
	\end{figure}
	
	\subsection{Divisibility Diagram of Tensor Product Dynamics}
	\label{2qubit_dd}
	
	We now investigate the structure of the divisibility diagram for tensor products ${\Phi_t^{\vb{p}}\otimes \Phi_t^{\vb{q}}}$. It will become evident that there exist regions $\mathcal{CP}_2$ of CP-divisible maps and $\mathcal{N}_2$ of non-P-divisible maps that exhibit SBFI. However, the two-qubit parameter space will also be complemented by a non-trivial region $(\mathcal{P}\setminus\mathcal{CP})_2$ of only P-divisible maps that do not show SBFI, as we shall prove in the following.
	\begin{itemize}
		\item \textit{{CP-divisibility.}} Since the tensor product of CP-divisible maps is still CP-divisible,
		$$
		\vb{p},\vb{q}\in \mathcal{CP}\Longrightarrow (\vb{p},\vb{q})\in \mathcal{CP}_2 \,.
		$$
		On the other hand, if $(\vb{p},\vb{q})\in \mathcal{CP}_2$,
		consider arbitrary rank-1 projectors $Q, P, R \in\Md{2}\otimes\Md{2}$; since $\Phi_{t,s}^{\vb{p}}\otimes\Phi_{t,s}^{\vb{q}}$ is a CPTP map for all $t \ge s \ge 0$,
		$$
		0\le \Tr\left((\mathds{1}_4^{1,3}\otimes Q^{2,4}) \,\Phi_{t,s}^{\vb{p}}\otimes\Phi_{t,s}^{\vb{q}}\otimes \mathrm{id}_4\left[R^{1,3}\otimes P^{2,4}\right]\right)=\Tr\left(Q\left(\Phi_{t,s}^{\vb{q}}\otimes\mathrm{id}_2\right)[P]\right)\,,
		$$
		(superscripts denote only to which qubits the $4\times 4$ matrices refer) showing that $\Phi_{t,s}^{\vb{q}}$ has to be completely positive. The same holds for $\Phi_{t,s}^{\vb{p}}$,
		so that
		
		$$
		\vb{p},\vb{q}\in \mathcal{CP}\Longleftarrow (\vb{p},\vb{q})\in \mathcal{CP}_2 \,.
		$$
		\item 
		\textit{{Lack of P-divisibility.}}
		In view of Proposition \ref{benattischrusfilippov_prop},
		$$
		\vb{p}\in \mathcal{P}\setminus\mathcal{CP}\Longrightarrow (\vb{p},\vb{p})\in \mathcal{N}_2 \,.
		$$
		Also, as we have seen, all single-dynamical maps labeled by $\vb{p}$ are P-divisible; thus, if $\Phi_t^{\vb{p}}\otimes \Phi_t^{\vb{p}}$ is not P-divisible, $\Phi_t^{\vb{p}}$ cannot be CP-divisible, and then
		$$
		\vb{p}\in \mathcal{P}\setminus\mathcal{CP}\Longleftarrow (\vb{p},\vb{p})\in \mathcal{N}_2 \,.
		$$
		Furthermore, for a sufficiently  small but not vanishing  perturbation $\delta \vb{p}$, $\norm{\delta \vb{p}}\ll 1$, one also has
		\begin{equation}
			\vb{p}\in \mathcal{P}\setminus\mathcal{CP} \Longrightarrow (\vb{p},\vb{p}+\delta \vb{p})\in\mathcal{N}_2 \,.
		\end{equation}
		Indeed, a small perturbation cannot in general restore the P-divisibility of the tensor product (see Remark \ref{remarkD1}). Moreover, we can also assert that
		\begin{equation}\label{pqInP}
			\vb{p},\vb{q}\in \mathcal{P}\setminus\mathcal{CP}\implies(\vb{p},\vb{q})\in\mathcal{N}_2 \,.
		\end{equation}
		Indeed, if both $\vb{p}$ and $\vb{q}$ are in $\mathcal{P}\setminus\mathcal{CP}$, then two rates, say, $\gamma_i^{\vb{p}}(t)$ and $\gamma_j^{\vb{q}}(t)$,
		would become negative and stay negative asymptotically:  $\gamma_i^{\vb{p}}(\infty)$, $\gamma_j^{\vb{q}}(\infty)<0$, by {Property}
		~(\emph{ii}). In particular, so will their sum, ${\gamma_i^{\vb{p}}(\infty)+ \gamma_j^{\vb{q}}(\infty)<0}$. Thus,  $\Phi_t^{\vb{p}}\otimes \Phi_t^{\vb{q}}$ cannot be P-divisible, due to Proposition~\ref{prop:qubit2tensordiff}.
		As a consequence, all  $\vb{p},\vb{q} \in \mathcal{P}\setminus\mathcal{CP}$ will give rise to a tensor product map $\Phi_t^{\vb{p}}\otimes\Phi_t^{\vb{q}}$ displaying SBFI.
		

		\item \textit{{P-divisibility without CP-divisibility.}} 
		We now prove the existence of a non-trivial region $(\mathcal{P}\setminus\mathcal{CP})_2$ of maps  $\Phi_t^{\vb{p}}\otimes \Phi_t^{\vb{q}}$
		which are  P-divisible, without $\Phi_t^{\vb{p}}$, $\Phi_t^{\vb{q}}$ being both CP-divisible.
		From \eqref{pqInP}, such a region can only consist of points $(\vb{p},\vb{q})$ with $\vb{p}\in\mathcal{P}\setminus\mathcal{CP}$ and $\vb{q}\in \mathcal{CP}$ or $\vb{p}\in\mathcal{CP}$ and $\vb{q}\in \mathcal{P}\setminus\mathcal{CP}$.
		In order to show that the region $(\mathcal{P}\setminus\mathcal{CP})_2$ is not empty, we restrict to $\vb{p}$ of the form
		$(p,p,1-2\,p)$, $0 \le p \le \frac{1}{2}$. These are the points lying along the line $p_1=p_2$ in Figure \ref{fig:regions}. With this choice, the Pauli rates become
		\begin{subequations}
			\begin{align}
				\gamma_1^{\vb{p}}(t)&=\gamma_2^{\vb{p}}(t)=\frac{2 p}{(1-2 p) e^{2 t}+2 p} \,,\\
				\gamma_3^{\vb{p}}(t)&=2\,\frac{1-p}{1+ p \left(e^{2 t}-1\right)}-\gamma_1^{\vb{p}}(t) \,.
			\end{align}
		\end{subequations}
		Thus, only $\gamma_3^{\vb{p}}(t)$ might become negative. If $\gamma_3^{\vb{p}}(t)\geq 0$, then  the inequality~\eqref{reg3} must be satisfied. For $p_1=p_2=p$, it reads
		$$
		2\, p \,  (p^2 - 3 p + 1) \ge 0\ .
		$$
		Then, $\vb{p}=(p,p,1-2p)\in \mathcal{CP}$ if and only if $0\le p\le p^*$, with $p^*=\frac{1}{2} \left(3-\sqrt{5}\right)\approx0.38$. Conversely, for $p^*< p\le \frac{1}{2}$, $\gamma_3^{\vb{p}}(t)$ is not a positive function of time and $\Phi_t^{\vb{p}}$ is not CP-divisible.
		
		\noindent
		Setting  $p^*<p\le\frac{1}{2}$ so that $\vb{p}=(p,p,1-2p)\in \mathcal{P}\setminus\mathcal{CP}$, we then look for parameters $\vb{q}\in\mathcal{CP}$ such that  $\Phi_t^{\vb{p}}\otimes \Phi_t^{\vb{q}}$ is P-divisible. From Proposition \ref{prop:qubit2tensordiff}, necessary and sufficient conditions are
		\begin{equation}\label{conditionpq}
			\gamma_3^{\vb{p}}(t)+\gamma_k^{\vb{q}}(t) \ge 0\,, \qquad \forall \,t\ge 0 \,,\quad  k=1,2,3\,.
		\end{equation}
		First, let us look for a $\vb{q}$ on the line $q_1 = q_2 \equiv q$, with $0\le q \le p^*$. Then, a simpler condition can be inferred, namely,
		\begin{equation}\label{conditionpq_asympt}
			\gamma_3^{\vb{p}}(\infty)+\gamma_k^{\vb{q}}(\infty) \ge 0\,\quad k=1,2,3\,.
		\end{equation}
		Indeed, $\gamma_3^{\vb{q}}(t)+\gamma_k^{\vb{q}}(t)$ can have at most one zero for $t\ge 0$ (see Appendix \ref{theappendix}). In turn, the following holds:
		\begin{equation}
			\exists\,  t^*\ge 0 : \gamma_3^{\vb{p}}(t^*)+\gamma_k^{\vb{q}}(t^*) <0 \implies 	\gamma_3^{\vb{p}}(t)+\gamma_k^{\vb{q}}(t) <0 \quad  \forall \,t>t^*\,,
		\end{equation}
		so that one can look only at the asymptotic behavior. Then, \eqref{conditionpq_asympt} enforces the following conditions on the pair  $(p,q)$:
		\begin{subequations}\label{bisectorconditions_alt}
			\begin{align}
				q (1 - 2 q) + p (1 - 6 q + 7 q^2) -
				p^2 (2 - 7 q + 4 q^2) \ge 0 \,,\\
				1-2 q- p (3 - 7 q)+p^2 (1 - 4 q)  \ge 0 \,.
			\end{align}
		\end{subequations}
		Hence, if $\vb{p}=(p,p,1-2p)\in\mathcal{P}\setminus\mathcal{CP}$ and $\vb{q}=(q,q,1-2q)\in\mathcal{CP}$  satisfy (46a) and (46b), then $(\vb{p},\vb{q})\in (\mathcal{P}\setminus\mathcal{CP})_2$; otherwise, they belong to $\mathcal{N}_2$. Regions defined by  (46a) and (46b) are displayed in Figure~\ref{fig:boundariespq}: it shows  that, for values of $p$ increasingly greater than $p^{*}$, the following are true: 
		\begin{enumerate}
			\item The largest value $0\le q\le p^*$ such that $\Phi_t^{\vb{p}}\otimes\Phi_t^{\vb{q}}$ is P-divisible
			is increasingly smaller than $p^*$ (and $p$). Thus, what also increases is
			the minimal departure \mbox{$\delta \vb{q}=\norm{\vb{q}-\vb{p}}$ of $\vb{q}$}
			from $\vb{p}$ in the second party such that $\Phi_t^{\vb{p}}\otimes\Phi_t^{\vb{p}+\delta\vb{q}}$ is \mbox{P-divisible;}
			\item The subset of $0\le q\le p^*$ such that $\Phi_t^{\vb{p}}\otimes\Phi_t^{\vb{q}}$ is P-divisible progressively reduces, up to $p=\sqrt{2}-1\approx0.414$, for which only ${q}=\frac{1}{3}$ can make $\Phi_t^{\vb{p}}\otimes \Phi_t^{\vb{q}}$ P-divisible (that is, the only $\vb{q}$ making the map P-divisible is the centroid of $\mathcal{CP}$). \end{enumerate}

		\begin{figure}[t]
			\centering
			\includegraphics[width=0.4\linewidth]{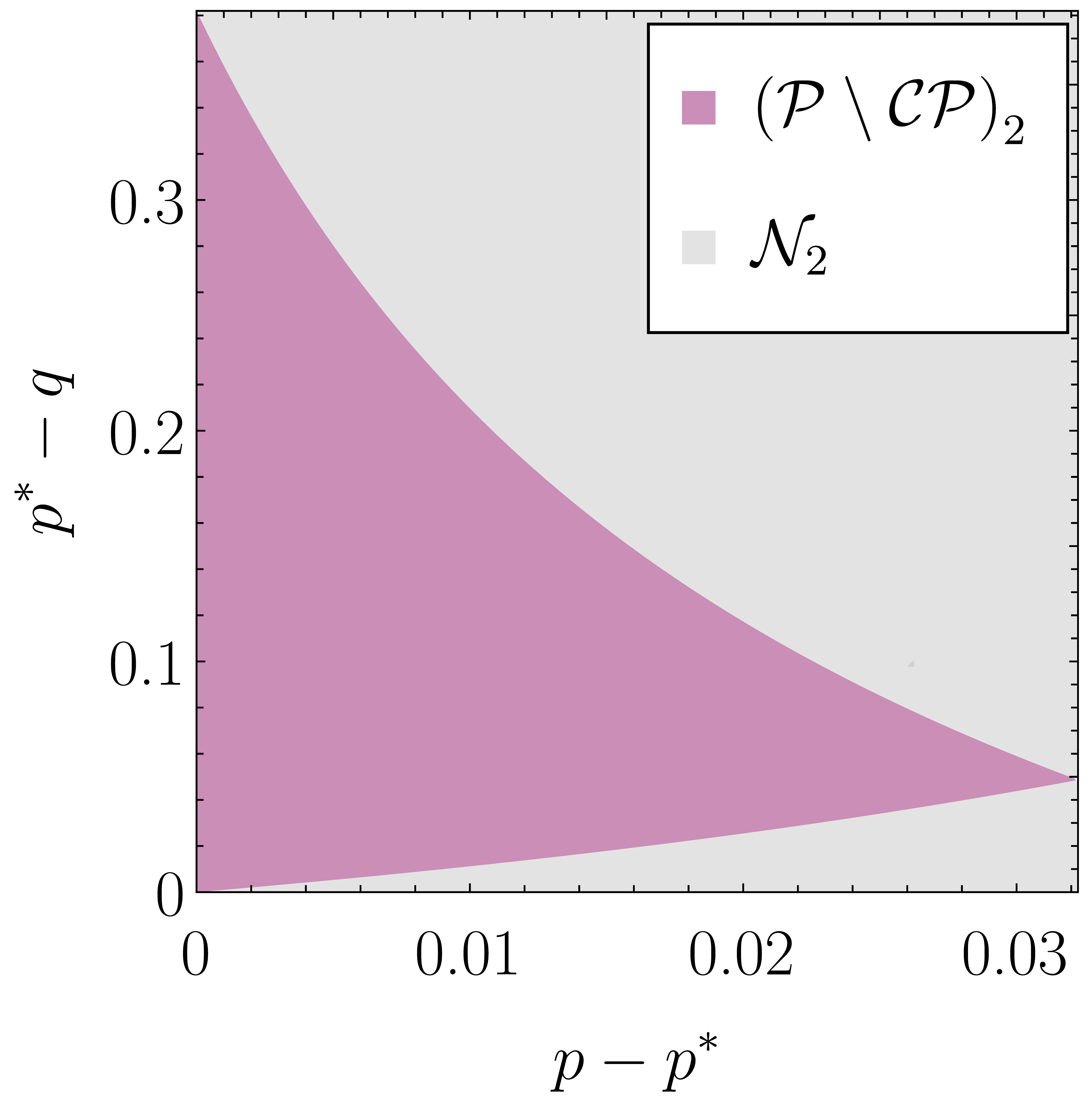}
			\caption{{Parameter set of maps $\Phi_t^{\vb{p}}\otimes \Phi_t^{\vb{q}}$, with $\vb{p}=(p,p,1-2p)$, $\vb{q}=(q,q,1-2q)$ along the bisector of Figure~\ref{fig:regions},
					with $p\in(p^*,\frac{1}{2}]$, $q\in[0,p^*]$ and $p^*=\frac{1}{2}(3-\sqrt{5})\approx 0.38$. The region of $(p,q)$ defined by the inequalities in (46a) and (46b) is shown in purple. For points inside this region, $\Phi_t^{\vb{p}}\otimes \Phi_t^{\vb{q}}$ is P-divisible, while outside it is not P-divisible and displays SBFI. Notice that for increasing $p>p^*$, the minimal distance $\delta q:=p-q$ needed to make $\Phi_t^{\vb{p}}\otimes \Phi_t^{\vb{p}}$ enter the purple region by varying it into $\Phi_t^{\vb{p}}\otimes \Phi_t^{\vb{q}}$ with
					$\vb{p}=(p,p,1-2p)$ and $\vb{q}=(p+\delta q,p+\delta q, 1-2p-2\delta q)$ increases.  The interval of ``good'' $q$ that permits us to restore P-divisibility also reduces with increasing $p$, up until $p=\sqrt{2}-1\approx0.414$, which corresponds to $p-p^*\approx 0.03$, where it only consists of  $q=\frac{1}{3}$. The corresponding point $\vb{q}=(\frac{1}{3},\frac{1}{3},\frac{1}{3})$ lies on the centroid of $\mathcal{CP}$.}}
			\label{fig:boundariespq}
		\end{figure}


		To have more insight about the structure of the region $(\mathcal{P}\setminus\mathcal{CP})_2$, let us now allow $q_1 \ne q_2$, while keeping $\vb{p} \in(\mathcal{P}\setminus\mathcal{CP})$ fixed. Define then a region
		\begin{equation}
			\left(\mathcal{P}\setminus\mathcal{CP}\right)_{2,t}^{\vb{p}}=\{\vb{q} \in \mathcal{CP}: \gamma_3^{\vb{p}}(t)+\gamma_k^{\vb{q}}(t)\ge 0,\, k=1,2,3\} \, .
		\end{equation}
		At $t=0$, $\gamma_3^{\vb{p}}(0)+ \gamma_k^{\vb{q}}(0)=2(p_3+q_k) \ge0$, so $(\mathcal{P}\setminus\mathcal{CP})_{2,t=0}^{\vb{p}}\equiv\mathcal{CP}$.
		In Figure~\ref{fig:P2}a, $(\mathcal{P}\setminus\mathcal{CP})^{\vb{p}}_{2,t}$ is depicted for  $\tilde{\vb{p}}=(0.4,0.4,0.2)$; it shrinks for increasing $t$, but it survives the limit $t \to \infty$. Points $\vb{q}$ within the asymptotic region $(\mathcal{P}\setminus\mathcal{CP})_{2,\infty}^{\tilde{\vb{p}}}\subseteq\mathcal{CP}$ will give rise to maps $\Phi_{t}^{\tilde{\vb{p}}}\otimes \Phi_{t}^{\vb{q}}$, which are P-divisible. Numerical evidence indeed shows that~\eqref{conditionpq_asympt} is still  necessary and sufficient for having the P-divisibility of the tensor product, since the sum $\gamma_3^{\vb{p}}(t)+\gamma_k^{\vb{q}}(t)$ will have at most one zero if $\vb{q}\in \mathcal{CP}$ ({{this can be checked} 
			by studying the numerator of $\gamma_3^{\vb{p}}(t)+\gamma_k^{\vb{q}}(t)$, which will be a polynomial in $e^{2t}$, and studying the pattern of coefficients to apply the Descartes rule of signs. This is carried out explicitly in Appendix~\ref{theappendix} for the case of $\vb{p},\vb{q}$ taken along the lines $p_1=p_2$ and $q_1=q_2$}).

		In addition, for $\vb{p}$ along the line $p_1=p_2$ and progressively further away from $(p^*,p^*,1-2p^*)$, the size of $(\mathcal{P}\setminus\mathcal{CP})_{2,\infty}^{\vb{p}}\subseteq\mathcal{CP}$ will reduce, as already noted for the $q_1=q_2$ case. 
		The regions $(\mathcal{P}\setminus\mathcal{CP})_{2,\infty}^{\vb{p}}$ for $\tilde{\vb{p}}=(0.4,0.4,0.2)$ and $\overline{\vb{p}}=(0.41,0.41,0.18)$ are compared in Figure~\ref{fig:P2}b.

	\end{itemize}

	\begin{Remark}
		Having fixed $\vb{p}\in (\mathcal{P}\setminus\mathcal{CP})$, the P-divisibility of $\Phi_t^{\vb{p}}\otimes \Phi_t^{\vb{q}}$ cannot be restored by making arbitrarily small perturbations to $\vb{p}$, as discussed in Remark \ref{remarkD1}. This is evident also from \mbox{Figures~\ref{fig:boundariespq} and \ref{fig:P2}b}, since it is necessary to perform a large enough variation in the second party to enter region $(\mathcal{P}\setminus\mathcal{CP})_{2,\infty}^{\vb{p}}$, restore P-divisibility and, consequently, eliminate SBFI.
	\end{Remark}\vspace{-12pt}
	
	\begin{figure}[t]
		\centering
		\captionsetup[subfloat]{justification=centering}
		\subfloat[\label{fig:P2times}]{\includegraphics[height=0.45\linewidth]{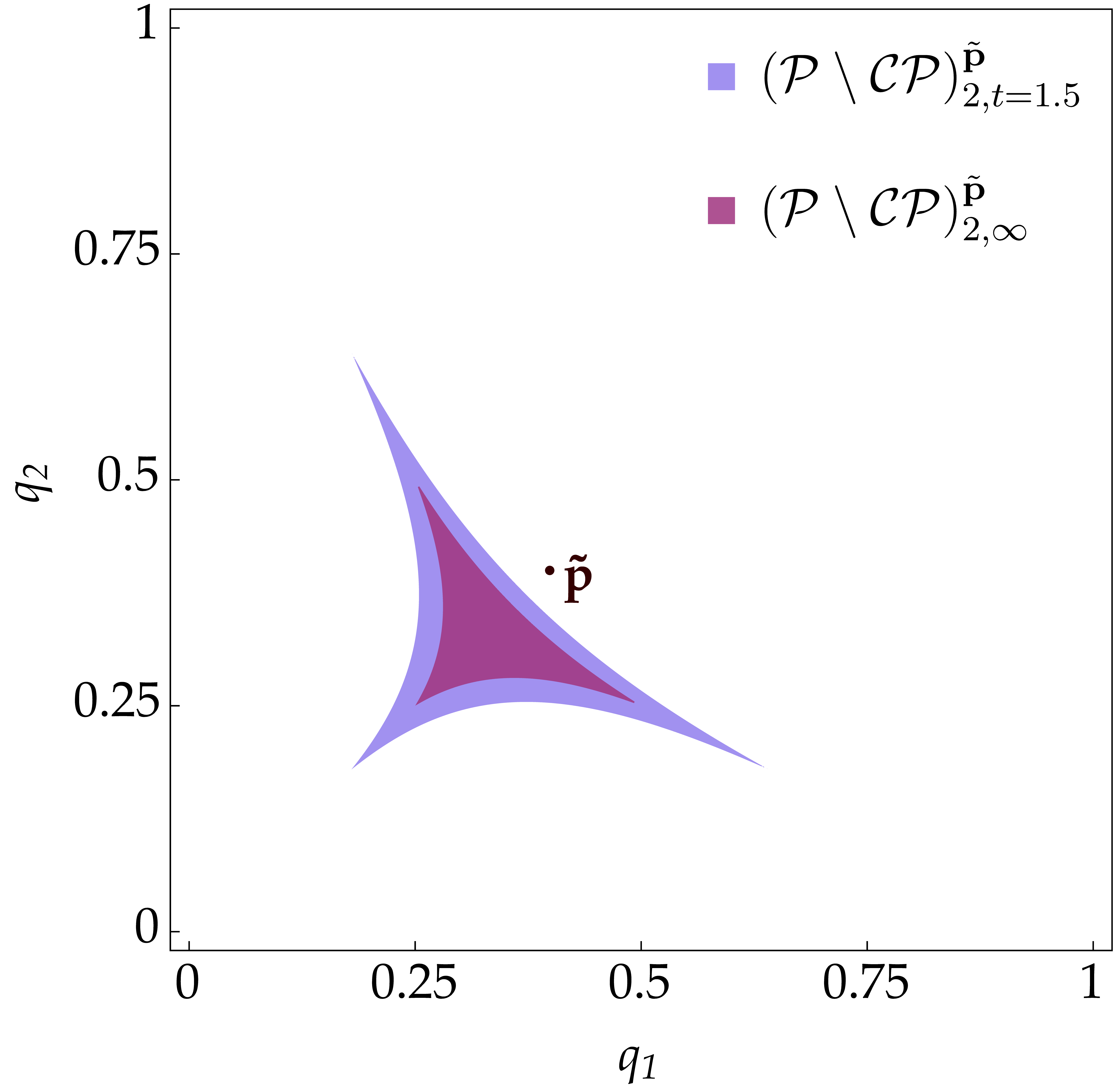}}%
		\vspace{-12pt}
		\subfloat[\label{fig:P2regions}]{\includegraphics[height=0.45\linewidth]{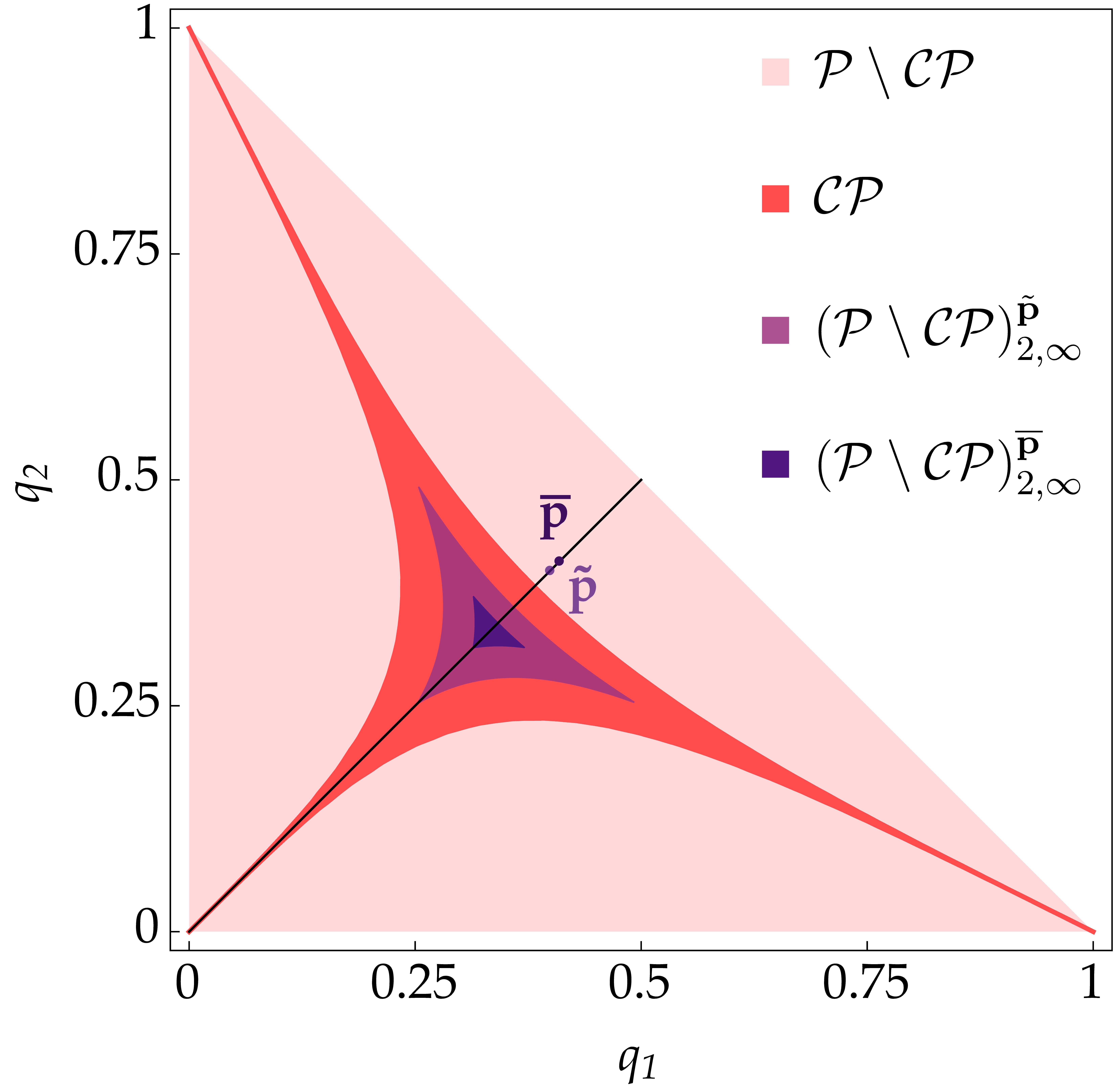}}%

		\caption{In (\textbf{a}), for fixed $\tilde{\vb{p}}=(0.4,0.4,0.2)$ in $(\mathcal{P}\setminus\mathcal{CP})$, the region $(\mathcal{P}\setminus\mathcal{CP})_{2,t}^{\tilde{\vb{p}}}=\{\vb{q} \in \mathcal{CP}: \gamma_3^{\tilde{\vb{p}}}(t)+\gamma_k^{\vb{q}}(t)\ge 0,\,k=1,2,3\}$ is displayed for $t=1.5$ and for $t\to \infty$. In (\textbf{b}), the same region $(\mathcal{P}\setminus\mathcal{CP})_{2,\infty}^{\tilde{\vb{p}}}$ is displayed in purple within $\mathcal{CP}$; it contains points $\vb{q}$ such that $\Phi_t^{\tilde{\vb{p}}}\otimes \Phi_t^{\vb{q}}$ is P-divisible. Notice that it cannot be reached from $\tilde{\vb{p}}$ with an arbitrarily small perturbation $\tilde{\vb{p}}+\delta \vb{p}$: a large enough  variation is required to enter it, restore P-divisibility and eliminate SBFI. The same region is also displayed in dark purple for $\overline{\vb{p}}=(0.41,0.41,0.18)$, illustrating that as $\vb{p}=(p,p,1-2p)$ moves along the $p_1=p_2$ line with $p$ increasingly larger than $p^*$, the size of the region $(\mathcal{P}\setminus\mathcal{CP})_{2,\infty}^{\vb{p}}$ decreases.}
		\label{fig:P2}
	\end{figure}

	\section{Conclusions}
	
	This paper deals with a peculiar quantum phenomenon arising within the scenario of open multi-partite quantum systems whose joint dynamics is the tensor product of dynamical maps with time-dependent generators. It might indeed happen that memory effects in single dynamics do not support BFI from environment to system, while the tensor product dynamics do. Such a superactivation phenomenon is a purely quantum effect resulting from non-classical correlations among the parties, as can be shown by means of a suitable collisional model~\cite{benattinichele}. Its physical appeal relies upon the fact that memory effects have often been proven resourceful in a variety of quantum technological tasks. In this respect, further investigations are needed for potential applications of SBFI. \par
	The absence or presence of BFI in connection to an open quantum dynamics $\Lambda_t$ is determined using its divisibility properties, namely, by whether or not the intertwining maps $\Lambda_{t,s}$ such that $\Lambda_t=\Lambda_{t,s}\circ\Lambda_s$ are positive or not.
	We investigated such properties for the general tensor products of completely positive dynamical maps $\Lambda_t^{(1)}\otimes\Lambda_t^{(2)}$, with the aim of seeking conditions for P-divisibility, without the single-party maps $\Lambda_t^{(1,2)}$ being CP-divisible, namely, without these maps having completely positive intertwiners.
	
	For the case of two P-divisible Pauli maps acting on qubits, particularly simple necessary and sufficient conditions for the P-divisibility of their tensor products have been given in terms of the mutual sums of the master equation rates. 
	Moreover, examples have been provided in which the P-divisibility of the tensor product could be achieved without one or even both maps being CP-divisible. The interest of these results relates to the fact that a lack of P-divisibility (along with invertibility) implies the emergence of BFI
	through revivals of the trace distance between time-evolving quantum states. If $\Lambda_t$ is P- but not CP-divisible, it is known that BFI does not occur for $\Lambda_t$, but it does occur for the second-order tensor product $\Lambda_t\otimes\Lambda_t$, describing two parties evolving in identical, dissipative environments. Such an intriguing phenomenon has been called the Superactivation of BFI.
	
	The  results presented in the manuscript imply that SBFI can be switched off by performing a suitable variation of one of the two dissipative evolutions. Concretely, by suitably changing the parameters of the generator of one of the parties, for instance, by acting on the microscopic origin of open quantum dynamics, one can stop information from being injected into the coupled open
	quantum system from the environment. 
	Abundant phenomenology relative to the SBFI effect has been provided by means of mixtures of dephasing qubit maps, in particular by means of two-qubit ``divisibility diagrams'' that show regions of CP-divisible maps, of non-P-divisible maps as well as of only P-divisible maps. The emerging picture is such that, when dealing with P-divisible but not CP-divisible dynamics $\Lambda_t$, obtaining a P-divisible tensor product $\Lambda_t\otimes\widetilde{\Lambda}_t$ is not achievable via simple tensorization with a slight perturbation $\widetilde{\Lambda}_t$ of $\Lambda_t$ itself; instead, a sufficiently large variation is required to eliminate SBFI.
	
	
	
	
	\vspace{6pt}
	
	
	
	

	F.B. and G.N. acknowledge financial support from PNRR MUR project PE0000023-NQSTI.

	
	
		
	

	\appendix

	\section*{Appendix}
	\renewcommand{\thesubsection}{\Alph{subsection}}
	\subsection{Proofs of Propositions \ref{prop:sufficientPD} and  \ref{prop:qubit2tensordiff}}\label{app:APPproofs}
	
	\begin{proof}[{Proof} of Proposition~\ref{prop:sufficientPD}] 
		Let $\braket{\phi}{\psi}=0$ and consider the quantity
		
		$$
		\mathcal{G}_t(\phi,\psi):=\bra{\phi}\mathcal{L}_t^{(1)}\otimes \mathrm{id}_d + \mathrm{id}_d \otimes \mathcal{L}_t^{(2)}\left[\ket{\psi}\bra{\psi}\right]\ket{\phi}\ .
		$$
		
		{Let} $\Phi=[\phi_{ab}]$, $\Psi=[\psi_{ab}]$
		be the matrices whose entries are the vectors' components with respect to a fixed orthonormal basis $\{\ket{a}\otimes\ket{b}\}_{a b}$, $\Tr(\Phi\daga{\Psi})=0$, so that:
		
		\begin{equation}
			\mathcal{G}_t(\phi,\psi)=\sum_{k=1}^{d^2-1} \gamma_k^{(1)}(t) \abs{\Tr(F_k^{(1)}(t) \, \Phi \Psi^{\dagger})}^2+\sum_{k=1}^{d^2-1} \gamma_k^{(2)}(t) \abs{\Tr(F_k^{(2)}(t) \, (\Psi^{\dagger}\Phi )^{T})}^2.
			\label{eq6}
		\end{equation}
		
		\textls[-25]{{Since} ${F_k^{(\alpha)}(t)=\left(F_k^{(\alpha)}(t)\right)^\dagger}$, $\alpha=1,2$, are assumed to be a Hilbert--Schmidt orthonormal basis for the traceless $d\times d$ matrices $W$ for all $t\ge0$, one writes $W=\sum_{j=1}^{d^2-1}{\rm Tr}\Big(WF^{(\alpha)}_j(t)\Big)\,F^{(\alpha)}_j(t)$} so that
		\begin{equation*}
			\begin{split}
				\sum_{k=1}^{d^2-1}\left(\Tr(F_k^{(1)}(t) \, \Phi \Psi^{\dagger})\right)^2&=\Tr(\Phi \daga{\Psi}\Phi \daga{\Psi})=\Tr((\daga{\Psi}\Phi)^T (\daga{\Psi}\Phi)^T)\\&=\sum_{k=1}^{d^2-1}\left(\Tr(F_k^{(2)}(t) \, (\daga{\Psi}\Phi)^T)\right)^2,
			\end{split}
		\end{equation*}
		from which, isolating one term and applying the triangle inequality, one obtains
		\begin{equation*}
			\left|\Tr(F_j^{(2)} \, (\Psi^{\dagger}\Phi )^{T})\right|^2\le \sum_{k=1}^{d^2-1}\left|\Tr(F_k^{(1)} \, \Phi \Psi^{\dagger})\right|^2 +\sum_{k \ne j}\left|\Tr(F_k^{(2)} \, (\Psi^{\dagger}\Phi )^{T}) \right|^2.
		\end{equation*}
		
		{Set} $J_-:=\{t\ge 0 : \gamma_j^{(2)}(t) <0\}\subseteq \mathbb{R}^+$.
		For all $t \in J_-$, we can thus bound~\eqref{eq6} as~follows:
		\begin{eqnarray*}
			\mathcal{G}_t(\phi,\psi)&\ge&\sum_{k=1}^{d^2-1} \left(\gamma_k^{(1)}(t)+\gamma_j^{(2)}(t)\right)\,
			\abs{\Tr(F_k^{(1)} \, \Phi \Psi^{\dagger})}^2\\
			&+&\sum_{k \ne j} \left(\gamma_k^{(2)}(t)+\gamma_j^{(2)}(t)\right)\, \abs{\Tr(F_k^{(2)} \, (\Psi^{\dagger}\Phi )^{T})}^2\,\ge0\,.
		\end{eqnarray*}
		
		{We} can then define $I_-=\{t\ge0 : \gamma_i^{(1)}(t)<0\}$ and repeat the estimation, this time acting on the term proportional to $\gamma_i^{(1)}(t)$. Thus, $\mathcal{G}_t(\phi,\psi)\ge 0$ for all $t \in I_- \cup J_-$ for times $t$ such that $\gamma_i^{(1)}(t)\ge0$ and $\gamma_j^{(2)}(t) \ge 0$, $\mathcal{G}_t(\phi, \psi) \ge 0$ follows trivially from \eqref{eq6}.
		Hence, we prove that $\mathcal{G}_t(\phi, \psi)\ge0$ for all $t \ge 0 $, which is sufficient to guarantee, by means of Lemma~\ref{chp:prop_kossak}, the P-divisibility of $\Lambda_{t}^{(1)}\otimes\Lambda_{t}^{(2)}$. \end{proof}
	
	\begin{proof}[{Proof} of Proposition~\ref{prop:qubit2tensordiff}] 
		For the ``if part'', the conditions in \eqref{pdtensor_condition} together with the necessary and sufficient conditions for the P-divisibility of the one-qubit Pauli maps (see Example \ref{ex:Paulimaps}), namely,
		\begin{equation}
			\gamma_i^{(\alpha)}(t)+\gamma_j^{(\alpha)}(t)\ge0\,, \qquad \alpha=1,2\ ,\quad
			\forall \,t\ge0\, ,\quad \forall i \ne j \,,
		\end{equation}
		fulfil all the requirements of Proposition \ref{prop:sufficientPD}; hence, they are sufficient for the P-divisibility of $\Lambda_t^{(1)}\otimes \Lambda_t^{(2)}$. \par
		For the ``only if'' part, let us first show that the P-divisibility of $\Lambda_{t}^{(1)}\otimes \Lambda_{t}^{(2)}$ implies the P-divisibility of $\Lambda_t^{(1,2)}$. Suppose, then, that $\forall$ $t \ge s \ge 0$,  $\Lambda_{t,s}^{(1)}\otimes \Lambda_{t,s}^{(2)}$ is a positive map, and consider generic rank-1 projectors $P$, $P'$ and $Q$ $\in \Md{2}$; since $\mathds{1}_2\otimes Q\ge0$ and $\Lambda_{t,s}^{(1,2)}$ preserve the trace,
		\begin{align}
			0\le	\Tr\left(\left(\mathds{1}_2\otimes Q\right)\,\Lambda_{t,s}^{(1)}\otimes \Lambda_{t,s}^{(2)}\left[P'\otimes P\right]\right) =
			\Tr(Q\,\Lambda_{t,s}^{(2)}[P]) \,,
		\end{align}
		so $\Lambda_{t,s}^{(2)}$ is a positive map for all $t \ge s \ge 0$ and  $\Lambda_t^{(2)}$ is P-divisible. Analogously, $\Lambda_{t}^{(1)}$ also has to be P-divisible. Therefore, only the necessity of the conditions in \eqref{pdtensor_condition} remains to be demonstrated. Consider the unitary operators
		\begin{equation}
			V^{kl}=\frac{\sigma_k+\sigma_l}{\sqrt{2}}=\left(V^{kl}\right)^\dag=\left(V^{kl}\right)^{-1}\ , \quad k \ne l \in\{1,2,3\}\ ,
		\end{equation}
		\textls[-15]{	such that $V^{kl}\, \sigma_i\, V^{kl}=\sum_{j=1}^3\mathcal{V}_{ij}^{kl}\sigma_j$ with $\mathcal{V}^{kl}$ a $3\times 3$ Hermitian matrix.
			Due to \mbox{Proposition~\ref{prop:necPd}}, $\Lambda_t^{(1)}\otimes \Lambda_t^{(2)}$ is P-divisible $\implies$ $K^{(1)}(t)+\mathcal{V}^{kl}K^{(2)}(t)\mathcal{V}^{kl}\ge 0$. For example, $\mathcal{V}^{23}=\begin{pmatrix}
				-1 & 0 &0 \\
				0 & 0 &1 \\
				0 & 1 &0 \\
			\end{pmatrix}$ yields}
		\begin{eqnarray*}
			0 & \le &	K^{(1)}(t)+\mathcal{V}^{23}  K^{(2)}(t)\mathcal{V}^{23} \\
			& =& \begin{pmatrix}
				\gamma_1^{(1)}(t) &  & \\
				& \gamma_2^{(1)}(t) &\\
				& &\gamma_3^{(1)}(t)  \\
			\end{pmatrix}+
			\begin{pmatrix}
				\gamma_1^{(2)}(t) &  & \\
				& \gamma_3^{(2)}(t) &\\
				& &\gamma_2^{(2)}(t)  \\
			\end{pmatrix}\,,
		\end{eqnarray*}
		thus enforcing $\gamma_1^{(1)}(t) +\gamma_1^{(2)}(t) \ge0, \gamma_2^{(1)}(t) +\gamma_3^{(2)}(t)\ge0, \gamma_3^{(1)}(t) +\gamma_2^{(2)}(t)\ge0$.
		Varying $k,l$, one obtains the complete set of conditions in \eqref{pdtensor_condition}.
	\end{proof}

	\subsection{Properties of the Quasi-ENM Dynamics \texorpdfstring{$\mathbf\Phi_\textbf{\emph{t}}^{\vb{p}}$}{\texttwoinferior}}
	\label{theappendix}
	\begin{enumerate}[label=(\textit{\Roman*}),leftmargin=2.3em,labelsep=1.mm]
		\item \textit{If $p_1 \,p_2\, p_3>0$ and $\gamma_k^{\vb{p}}(t^*)<0$ for some $t^*>0$, $k\in\{1,2,3\}$, then $\gamma_k^{\vb{p}}(t)<0$ $\forall \, t > t^*$.}
		\begin{proof} Suppose that $\gamma_3^{\vb{p}}(t)$ is the rate turning negative at some $t^*>0$. Then, by means of Descartes's sign rule~\cite{Meserve1982}, we show that $\gamma_3^{\vb{p}}(t)$ can have at most one zero, so if it turns negative, it stays negative for all subsequent times.
			Using~\eqref{mixrates} and~\eqref{g-aux}, we~rewrite
			\begin{equation}\label{rewr_3}
				\gamma_3^{\vb{p}}(t)=\frac{N_3^{\vb{p}}(t)}{\prod_{k=1}^3 (1+p_k(e^{2t}-1))},
			\end{equation}
			where
			\begin{align}
				N_3^{\vb{p}}(t)\equiv \alpha_{3,0}({\vb{p}})+\alpha_{3,1}({\vb{p}})e^{2 t}+ \alpha_{3,2}({\vb{p}})e^{4t},
			\end{align}
			with
			\begin{subequations}\label{coeff_alpha}
				\begin{align}
					\alpha_{3,2}({\vb{p}})&= -p_1^2 p_2 - p_1 p_2^2 + p_1^2 p_3 + p_2^2 p_3 + p_1 p_3^2 + p_2 p_3^2\,,\\
					\alpha_{3,1}({\vb{p}})&=2 p_3 (1-p_2) (1-p_1)\,,\\
					\alpha_{3,0}({\vb{p}})&=(1-p_1) (1-p_2) (1-p_3)\,.
				\end{align}
			\end{subequations}
			
			{Thus}, $\alpha_{3,0}({\vb{p}})\ge 0$ and $\alpha_{3,1}({\vb{p}})\ge 0$ for all $\vb{p}$, while $\alpha_{3,2}({\vb{p}})$ can turn negative. By the Descartes rule of signs ({{recall that the rule states that for} 
				a polynomial with real coefficients,  the number of its positive roots $P$ is given by $P=S-2 m$, where $S$ is the number of sign changes between its nonzero coefficients  ordered in descending order of the powers, and $m\in \mathbb{N}$. In particular, if the number of sign changes is $S=1$ or $S=0$, there are, respectively, one or zero roots)}, $\gamma_3^{\vb{p}}(t)$ will have only one zero when $\alpha_{3,2}({\vb{p}})$ turns negative. Points $\vb{p}$ for which $\alpha_{3,2}({\vb{p}})$ turn negative are shown in Figure~\ref{figAPP:signs}a.
		\end{proof}
		
				%
		
		\begin{figure}[t]
			\centering
			\captionsetup[subfloat]{justification=centering}
			\subfloat[\label{fig:negsign}]{\includegraphics[height=5.5cm]{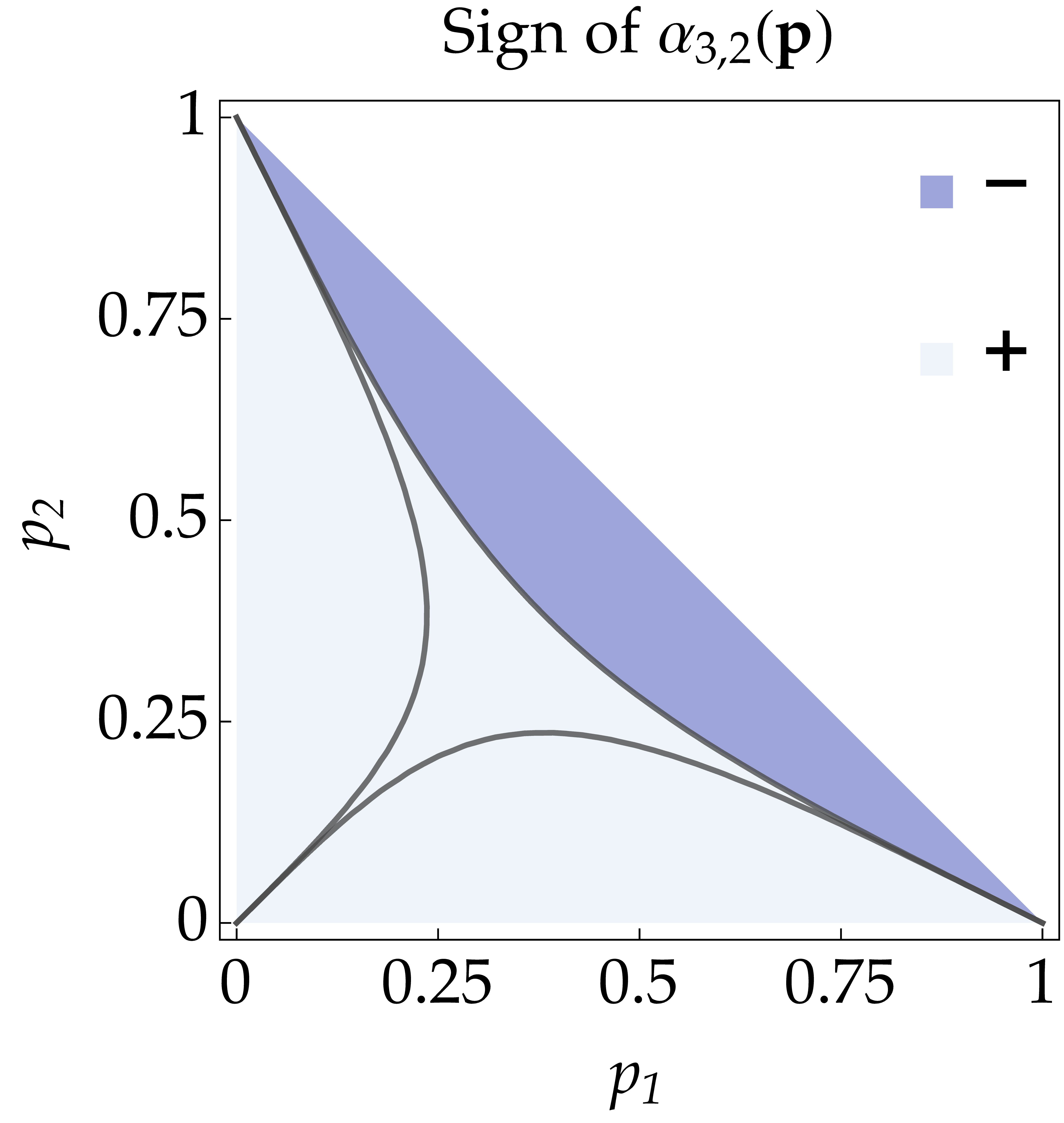}} \hspace{22pt}
			\subfloat[\label{fig:signpatterns33}]{\includegraphics[height=5.5cm]{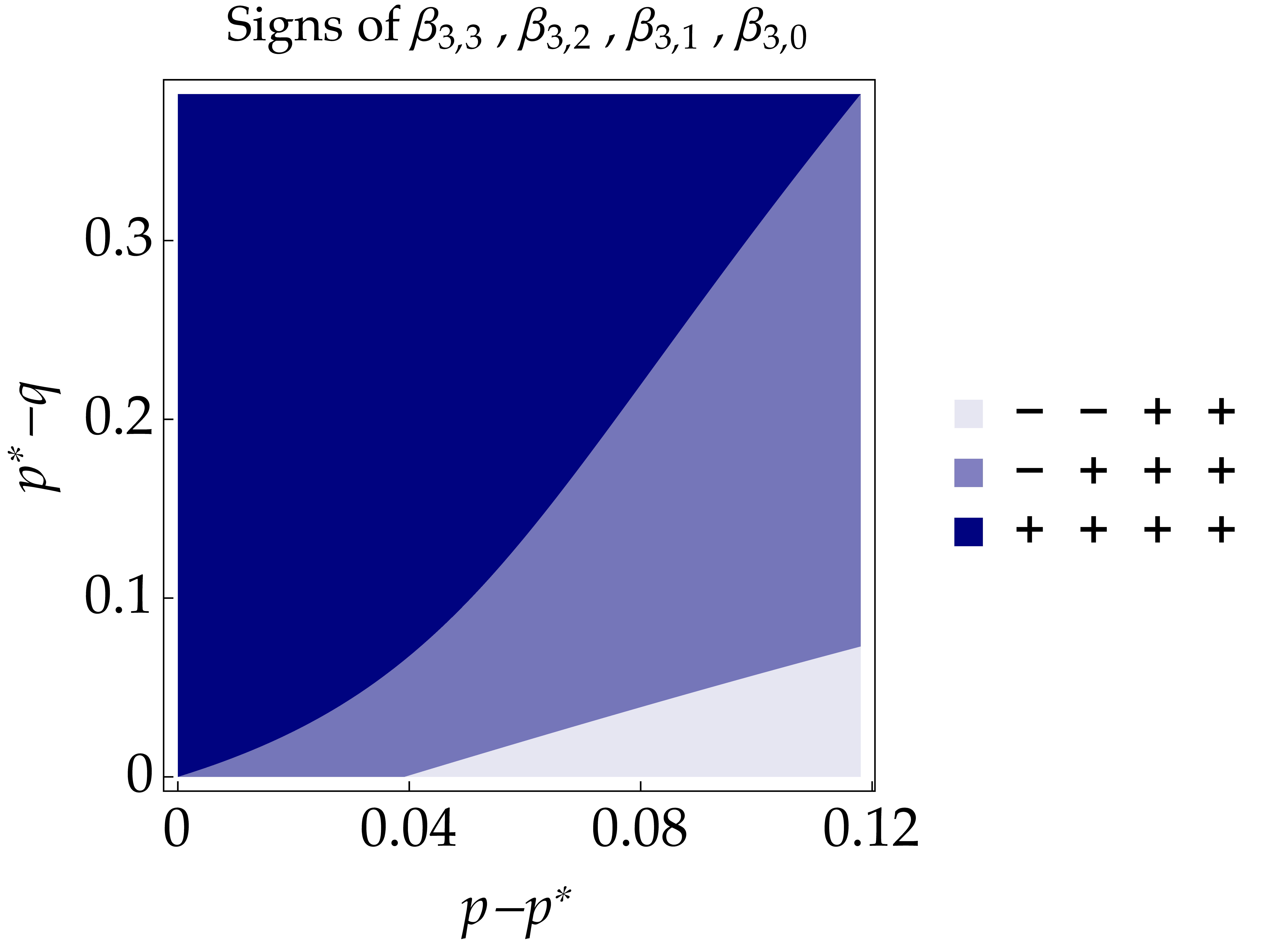}}%
			
			\caption{In (\textbf{a}), the subset of $\vb{p}$ that makes $\alpha_{3,2}(\vb{p})$  negative corresponds to the subregion of $\mathcal{P}$ for which $\gamma_3^{\vb{p}}(t)\not\ge0$. In (\textbf{b}), the sign patterns of $\beta_{3,3}$, $\beta_{3,3}$, $\beta_{3,1}$, $\beta_{3,0}$ are functions of $(p,q)$, $p\in(p^*,\frac{1}{2}]$ and $q\in[0,p^*]$, $p^*=\frac{1}{2}(3-\sqrt{5})$. Since there is only one sign flip, $\gamma_3^{\vb{p}}(t)+\gamma_3^{\vb{q}}(t)$ can then have up to one zero.}
			\label{figAPP:signs}
		\end{figure}

		\item \textit{Let $\vb{p}=(p,p,1-2p)$ and $\vb{q}=(q,q,1-2q)$, with $p^*< p \le 1/2$ and $ 0 \le q \le p^*$, $p^*=\frac{1}{2}(3-\sqrt{5})$. If  $\gamma_3^{\vb{p}}(t)+\gamma_k^{\vb{q}}(t)<0$ for some $t>t^*$, $k\in\{1,2,3\}$, then $\gamma_3^{\vb{p}}(t)+\gamma_k^{\vb{q}}(t)<0$ for all $t>t^*$}.
		\begin{proof}
			The reasoning is the same as in (\textit{I}). Let
			\begin{equation}
				\gamma_3^{\vb{p}}(t) + \gamma_k^{\vb{q}}(t)=\frac{N_{k}^{{p,q}}(t)}{D_k^{p,q}(t)}\,,
			\end{equation}
			where $D_k^{p,q}(t) \ge 0$  and
			\begin{align}
				N_{k}^{{p,q}}(t)=\sum_{n=0}^{n_k} \beta_{k,n}({{p,q}}) \,e^{2nt}\,.
			\end{align}
			
			Focusing only on the numerator, for $k=3$, one has $n_3=3$, and coefficients 	$\beta_{3,n}(p,q)$ read
			\begin{subequations}\label{coeff_beta3}
				\begin{align}
					\beta_{3,3}(p,q)&=2  \,(q (1 - 2 q) + p (1 - 6 q + 7 q^2) -
					p^2 (2 - 7 q + 4 q^2)) \,, \\
					\beta_{3,2}(p,q)&=2 \,(2 - 6 q + 5 q^2 - 2 p (3 - 10 q + 9 q^2) +
					p^2 (5 - 18 q + 12 q^2)) \,,\\
					\beta_{3,1}(p,q)&=
					6 \, (1 - q) (1 - p) (q + p (1 - 4 q)) \,,\\
					\beta_{3,0}(p,q)&=8 \,q\, (1 - q) \, p \,(1 - p)  \,.
				\end{align}
			\end{subequations}
			
			Coefficients $\beta_{3,0}(p,q)$ and $\beta_{3,1}(p,q)$ are always positive, while $\beta_{3,2}(p,q)$ and  $\beta_{3,3}(p,q)$ may become negative. Nevertheless, there can be at most one change of sign in the pattern of coefficients (considered in increasing order with the powers of $e^{2t}$), as shown in Figure~\ref{figAPP:signs}b. Again, by the Descartes rule of signs, one concludes that $\gamma_3^{\vb{p}}(t) + \gamma_3^{\vb{q}}(t)$ can have at most one real zero. The same reasoning applies for $\gamma_3^{\vb{p}}(t) + \gamma_1^{\vb{q}}(t)=\gamma_3^{\vb{p}}(t) + \gamma_2^{\vb{q}}(t)$, for which $n_1=n_2=2$ and
			\begin{subequations}\label{coeff_beta12}
				\begin{align}
					\beta_{i,2}(p,q)&=2 ( 1-2 q- p (3 - 7 q)+p^2 (1 - 4 q)) \,, \\
					\beta_{i,1}(p,q)&=2 (1 - p) ( 3 q+p (1 - 8 q)) \,,\\
					\beta_{i,0}(p,q)&=8 \,q \, p\, (1 - p) \,,
				\end{align}
			\end{subequations}
			where $i=1,2$,
			so that $\gamma_3^{\vb{p}}(t) + \gamma_2^{\vb{q}}(t)$ can have at most one zero for $t\ge 0$, depending on the sign of $\beta_{i,2}(p,q)$.
		\end{proof}
	\end{enumerate}


\begin{thebibliography}{99}
			
			\bibitem{GoriniKossSud}
			Gorini, V.; Kossakowski, A.; Sudarshan, E.C.G.
			\newblock {Completely positive dynamical semigroups of N‐level systems.}
			\newblock {\em J. Math. Phys.} {\bf 1976}, {\em 17},~821--825. [\href{http://doi.org/10.1063/1.522979}{CrossRef}]
			
			\bibitem{LindbladTh}
			{Lindblad, G.}
			\newblock {On the generators of quantum dynamical semigroups.}
			\newblock {\em Commun. Math. Phys.} {\bf 1976}, {\em 48},~119. [\href{http://dx.doi.org/10.1007/BF01608499}{CrossRef}]
			
			\bibitem{Davies}
			Davies, E.B.
			\newblock Markovian master equations.
			\newblock {\em Commun. Math. Phys.} {\bf 1974}, {\em
				39},~91--110. [\href{http://dx.doi.org/10.1007/BF01608389}{CrossRef}]
			
			\bibitem{gorini1976singular}
			Gorini, V.; Kossakowski, A.
			\newblock N-level system in contact with a singular reservoir.
			\newblock {\em J. Math. Phys.} {\bf 1976}, {\em
				17},~1298--1305. [\href{http://dx.doi.org/10.1063/1.523057}{CrossRef}]
			
			\bibitem{dumcke1985low}
			D{\"u}mcke, R.
			\newblock The low density limit for an N-level system interacting with a free
			Bose or Fermi gas.
			\newblock {\em Commun. Math. Phys.} {\bf 1985}, {\em
				97},~331--359. [\href{http://dx.doi.org/10.1007/BF01213401}{CrossRef}]
			
			\bibitem{WhatGoodFor}
			Li, C.F.; Guo, G.C.; Piilo, J.
			\newblock Non-Markovian quantum dynamics: What is it good for?
			\newblock {\em Europhys. Lett.} {\bf 2020}, {\em 128},~30001. [\href{http://dx.doi.org/10.1209/0295-5075/128/30001}{CrossRef}]
			
			\bibitem{bylickaNonMarkovianityReservoirMemory2014}
			Bylicka, B.; Chru{\'s}ci{\'n}ski, D.; Maniscalco, S.
			\newblock Non-{{Markovianity}} and Reservoir Memory of Quantum Channels: A
			Quantum Information Theory Perspective.
			\newblock {\em Sci. Rep.} {\bf 2014}, {\em 4},~5720. [\href{http://dx.doi.org/10.1038/srep05720}{CrossRef}]
			
			\bibitem{HuelgaPleniononMarkovianQuantumMetrology}
			Chin, A.W.; Huelga, S.F.; Plenio, M.B.
			\newblock Quantum Metrology in Non-Markovian Environments.
			\newblock {\em Phys. Rev. Lett.} {\bf 2012}, {\em 109},~233601. [\href{http://dx.doi.org/10.1103/PhysRevLett.109.233601}{CrossRef}]
			
			\bibitem{teleportationnonmrk}
			Laine, E.M.; Breuer, H.P.; Piilo, J.
			\newblock Nonlocal memory effects allow perfect teleportation with mixed
			states.
			\newblock {\em Sci. Rep.} {\bf 2014}, {\em 4},~4620. [\href{http://dx.doi.org/10.1038/srep04620}{CrossRef}]
			
			\bibitem{ChrusReview22}
			Chru{\'s}ci{\'n}ski, D.
			\newblock Dynamical maps beyond Markovian regime.
			\newblock {\em Phys. Rep.} {\bf 2022}, {\em 992},~1--85. [\href{http://dx.doi.org/10.1016/j.physrep.2022.09.003}{CrossRef}]
			
			\bibitem{BLP}
			Breuer, H.P.; Laine, E.M.; Piilo, J.
			\newblock Measure for the Degree of Non-{M}arkovian Behavior of Quantum
			Processes in {O}pen {S}ystems.
			\newblock {\em Phys. Rev. Lett.} {\bf 2009}, {\em 103},~210401. [\href{http://dx.doi.org/10.1103/PhysRevLett.103.210401}{CrossRef}] [\href{http://www.ncbi.nlm.nih.gov/pubmed/20366019}{PubMed}]
			
			\bibitem{BenattiChrusFil}
			Benatti, F.; Chru{\'s}ci{\'n}ski, D.; Filippov, S.
			\newblock Tensor power of dynamical maps and positive versus completely
			positive divisibility.
			\newblock {\em Phys. Rev. A} {\bf 2017}, {\em 95},~012112. [\href{http://dx.doi.org/10.1103/PhysRevA.95.012112}{CrossRef}]
			
			\bibitem{hastingsSuperadditivityCommunicationCapacity2009}
			Hastings, M.B.
			\newblock Superadditivity of Communication Capacity Using Entangled Inputs.
			\newblock {\em Nat. Phys.} {\bf 2009}, {\em 5},~255--257. [\href{http://dx.doi.org/10.1038/nphys1224}{CrossRef}]
			
			\bibitem{SmithQCapacity}
			Smith, G.; Yard, J.
			\newblock Quantum communication with zero-capacity channels.
			\newblock {\em Science} {\bf 2008}, {\em 321},~1812--1815. [\href{http://dx.doi.org/10.1126/science.1162242}{CrossRef}] [\href{http://www.ncbi.nlm.nih.gov/pubmed/18719249}{PubMed}]
			
			\bibitem{benattinichele}
			Benatti, F.; Nichele, G.
			\newblock Superactivation of Backflow of Information: Microscopic derivation.
			{2024},
			\textit{in press.}
			
			\bibitem{BreuerPetruccione}
			Breuer, H.P.; Petruccione, F.
			\newblock {\em The {T}heory of {O}pen {Q}uantum {S}ystem}; Oxford University
			Press:  {Oxford, UK,} 
			2002.
			
			\bibitem{RivasHuelgaPlenio}
			Rivas, {\'A}.; Huelga, S.F.; Plenio, M.B.
			\newblock Quantum non-{M}arkovianity: Characterization, quantification and
			detection.
			\newblock {\em Rep. Prog. Phys.} {\bf 2014}, {\em 77},~094001. [\href{http://dx.doi.org/10.1088/0034-4885/77/9/094001}{CrossRef}]
			
			\bibitem{RHPmeasure}
			Rivas, {\'A}.; Huelga, S.F.; Plenio, M.B.
			\newblock Entanglement and Non-{M}arkovianity of Quantum Evolutions.
			\newblock {\em Phys. Rev. Lett.} {\bf 2010}, {\em 105},~050403. [\href{http://dx.doi.org/10.1103/PhysRevLett.105.050403}{CrossRef}]
			
			\bibitem{Kossakowski}
			Kossakowski, A.
			\newblock On Necessary and sufficient conditions for a generator of a quantum
			dynamical semi-group.
			\newblock {\em Bull. Acad. Pol. Sci. Ser. Math. Astr. Phys.} {\bf 1972}, {\em
				20},~1021.
			
			\bibitem{ChrusWud2013}
			Chru{\'s}ci{\'n}ski, D.; Wudarski, F.A.
			\newblock Non-{M}arkovian random unitary qubit dynamics.
			\newblock {\em Phys. Lett. A} {\bf 2013}, {\em 377},~1425--1429. [\href{http://dx.doi.org/10.1016/j.physleta.2013.04.020}{CrossRef}]
			
			\bibitem{ChrusWud2015}
			Chru{\'s}ci{\'n}ski, D.; Wudarski, F.A.
			\newblock Non-{M}arkovianity degree for random unitary evolution.
			\newblock {\em Phys. Rev. A} {\bf 2015}, {\em 91},~012104. [\href{http://dx.doi.org/10.1103/PhysRevA.91.012104}{CrossRef}]
			
			\bibitem{Choi75}
			Choi, M.D.
			\newblock Completely positive linear maps on complex matrices.
			\newblock {\em Linear Algebra Its Appl.} {\bf 1975}, {\em
				10},~285--290. [\href{http://dx.doi.org/10.1016/0024-3795(75)90075-0}{CrossRef}]
			
			\bibitem{benattifloreaniniromano2002}
			\textls[-15]{Benatti, F.; Floreanini, R.; Romano, R.
				\newblock Complete positivity and dissipative factorized dynamics.
				\newblock {\em J. Phys. A Math. Gen.} {\bf 2002}, {\em
					35},~L551.} [\href{http://dx.doi.org/10.1088/0305-4470/35/39/101}{CrossRef}]
			
			\bibitem{GTDWissmannBreuerAmato}
			Wissmann, S.; Breuer, H.P.; Vacchini, B.
			\newblock Generalized trace-distance measure connecting quantum and classical
			non-{M}arkovianity.
			\newblock {\em Phys. Rev. A} {\bf 2015}, {\em 92},~042108. [\href{http://dx.doi.org/10.1103/PhysRevA.92.042108}{CrossRef}]
			
			\bibitem{ChrusManiscalco}
			Chru{\'s}ci{\'n}ski, D.; Maniscalco, S.
			\newblock Degree of Non-{M}arkovianity of Quantum Evolution.
			\newblock {\em Phys. Rev. Lett.} {\bf 2014}, {\em 112},~120404. [\href{http://dx.doi.org/10.1103/PhysRevLett.112.120404}{CrossRef}] [\href{http://www.ncbi.nlm.nih.gov/pubmed/24724632}{PubMed}]
			
			\bibitem{HallCanonical}
			Hall, M.J.W.; Cresser, J.D.; Li, L.; Andersson, E.
			\newblock Canonical form of master equations and characterization of
			non-{M}arkovianity.
			\newblock {\em Phys. Rev. A} {\bf 2014}, {\em 89},~042120. [\href{http://dx.doi.org/10.1103/PhysRevA.89.042120}{CrossRef}]
			
			\bibitem{CKR}
			Chru{\'s}ci{\'n}ski, D.; Kossakowski, A.; Rivas, {\'A}.
			\newblock Measures of non-{M}arkovianity: Divisibility versus backflow of
			information.
			\newblock {\em Phys. Rev. A} {\bf 2011}, {\em 83},~052128. [\href{http://dx.doi.org/10.1103/PhysRevA.83.052128}{CrossRef}]
			
			\bibitem{BenattiNONDEC}
			Benatti, F.; Floreanini, R.; Piani, M.
			\newblock Non-Decomposable Quantum Dynamical Semigroups and Bound Entangled
			States.
			\newblock {\em Open Syst. Inf. Dyn.} {\bf 2004}, {\em
				11},~325--338. [\href{http://dx.doi.org/10.1007/s11080-004-6622-6}{CrossRef}]
			
			\bibitem{BFP2004}
			Benatti, F.; Floreanini, R.; Piani, M.
			\newblock Quantum dynamical semigroups and non-decomposable positive maps.
			\newblock {\em Phys. Lett. A} {\bf 2004}, {\em 326},~187--198. [\href{http://dx.doi.org/10.1016/j.physleta.2004.04.046}{CrossRef}]
			
			\bibitem{Megier_etal}
			Megier, N.; Chru{\'s}ci{\'n}ski, D.; Piilo, J.; Strunz, W.T.
			\newblock Eternal non-{M}arkovianity: From random unitary to {M}arkov chain
			realisations.
			\newblock {\em Sci. Rep.} {\bf 2017}, \emph{7}, {\em 6379}. [\href{http://dx.doi.org/10.1038/s41598-017-06059-5}{CrossRef}]
			
			\bibitem{divdiagram2015}
			Chen, H.B.; Lien, J.Y.; Chen, G.Y.; Chen, Y.N.
			\newblock Hierarchy of Non-{{Markovianity}} and k -Divisibility Phase Diagram
			of Quantum Processes in Open Systems.
			\newblock {\em Phys. Rev. A} {\bf 2015}, {\em 92},~042105. [\href{http://dx.doi.org/10.1103/PhysRevA.92.042105}{CrossRef}]
			
			\bibitem{Meserve1982}
			\textls[-15]{Meserve, B.
				\newblock {\em Fundamental Concepts of Algebra}; Dover Books on Advanced
				Mathematics; Dover Publications: {Mineola, NY, USA},  1982.}
			
		\end{thebibliography}
\end{document}